\documentclass[10pt,draftcls,onecolumn]{IEEEtran}
\usepackage{epsfig,amssymb,latexsym,color,amsmath,pifont,colordvi,multicol}
\usepackage{subfigure}
\usepackage{times}
\definecolor{backgrey}{rgb}{0.86,0.86,0.86}
\definecolor{dblue}{rgb}{0,0.0,0.5}
\definecolor{dred}{rgb}{0.4,0.2,0}
\definecolor{dgreen}{rgb}{0.0,0.5,0}

\renewcommand{\baselinestretch}{2}

\newcommand{\captionfonts}{\small}
\makeatletter  % Allow the use of @ in command names
\long\def\@makecaption#1#2{%
  \vskip\abovecaptionskip
  \sbox\@tempboxa{{\captionfonts #1: #2}}%
  \ifdim \wd\@tempboxa >\hsize
    {\captionfonts #1: #2\par}
  \else
    \hbox to\hsize{\hfil\box\@tempboxa\hfil}%
  \fi
  \vskip\belowcaptionskip}
\makeatother   % Cancel the effect of \makeatletter
\newtheorem{theorem}{Theorem}

\newtheorem{assumption}{Assumption}

\newtheorem{remark}{Remark}
\newtheorem{proof}{Proof}
\newtheorem{example}{Example}

\newtheorem{corollary}{Corollary}

\newtheorem{definition}{Definition}

\newtheorem{lemma}{Lemma}

\title{\LARGE \bf The Information Flow and Capacity of Channels with Noisy Feedback}

\author{\quad Chong Li and Nicola Elia %\quad Arvind U. Raghunathan %
\thanks{Chong Li is a Ph.D. student with the Department of Electrical and
Computer Engineering, Iowa State University, Ames, IA, 50011
chongli@iastate.edu.}
\thanks{Dr.~Nicola Elia is with the Department of Electrical and
Computer Engineering,
Iowa State University,
Ames, IA, 50011 nelia@iastate.edu.}
\thanks{Part of the material in this paper has appeared in \cite{chong_isit11}.}%
}

\begin{document}
\maketitle \thispagestyle{empty} \pagestyle{plain}
\begin{abstract}
In this paper, we consider some long-standing problems in communication systems with access to noisy feedback. We introduce a new notion, the \textit{residual directed information}, to capture the effective information flow (i.e. mutual information between the message and the channel outputs) in the forward channel. In light of this new concept, we investigate discrete memoryless channels (DMC) with noisy feedback and prove that the noisy feedback capacity is not achievable by using any typical closed-loop encoder (non-trivially taking feedback information to produce channel inputs). We then show that the residual directed information can be used to characterize the capacity of channels with noisy feedback. Finally, we provide computable bounds on the noisy feedback capacity, which are characterized by the \textit{causal conditional directed information}.
\end{abstract}

\section{Introduction}
\indent The theory of feedback has been well studied \cite{Doyle92} for control systems but only partially investigated for communication systems. So far, a large body of work has looked at communication channels with perfect feedback and obtained many notable results. See \cite{Schalkwijk66,Schalkwijk66_2,cover89,Massey1990,Liu05,Ofer07,Tati09,Permuter09} and reference therein. As an illustration, it is known that perfect feedback improves the error exponent and reduces the coding complexity \cite{Cover88}. For channels with memory, using perfect feedback can increase the capacity compared with the non-feedback case \cite{cover89}. However, only few papers have studied channels with noisy feedback and many challenging problems are still open. Namely, how does noisy feedback affect the transmission rate in forward communication channels? Is noisy feedback helpful in improving decoding error exponent or reducing encoding complexity? More generally, is feedback beneficial to
communicate even though it is noisy? These questions are difficult because the noisy feedback induces a loss of coordination between the transmitter and the receiver. We can classify the results in the literature into two main categories. The first category studies the usefulness of noisy feedback by investigating reliability functions and error exponents. \cite{Draper06} shows that the noisy feedback can improve the communication reliability by specifying a variable-length coding strategy. \cite{Kim07} derives the upper and lower bounds on the reliability function of the additive white Gaussian noise channel with additive white Gaussian noise feedback. \cite{Burnashev08} considers a binary symmetric channel with a binary symmetric feedback link and shows that the achievable error exponent is improved under certain conditions. The second category focuses on the derivation of coding schemes mostly for additive Gaussian channels with noisy feedback based on the well-known Schalkwijk-Kailath scheme \cite{Schalkwijk66}. We refer interested readers to \cite{Omura68,Lavenberg71,Martins08,Chance10,Kumar} for details.\\
\indent Instead of concentrating on specific aspects or channels, in this paper, we study the noisy feedback problem in generality. We first focus on the effective information flow through channels with noisy feedback. We introduce a new concept, the residual directed information, which exactly quantifies the effective information flow through the channel and provides us with a clear view of the information flow in the noisy feedback channel. In light of this new concept, we show the failure of using the \textit{directed information} defined by Massey \cite{Massey1990} in noisy feedback channels, which is otherwise useful in the perfect feedback case. Furthermore, we investigate the DMC with \textit{typical noisy feedback} (definition \ref{def_typcialnoise}) and prove that the capacity is not achievable by using any typical closed-loop encoder (definition \ref{def_typicalencoder}). In other words, no encoder that typically (to be made more precise in the paper) uses the feedback information can achieve the capacity. This negative result is due to the fact that, by typically using noisy feedback, we need sacrifice certain rate for signaling in order to rebuild the cooperation of the transmitter and receiver such that the message can be recovered with arbitrarily small probability of error. Next, we give a general channel coding theorem in terms of the residual directed information for channels with noisy feedback, which is an extension of \cite{bookhan03}. The main idea is to convert the channel coding problem with noisy feedback into an equivalent channel coding problem without feedback by considering code-functions instead of code-words \cite{Shannon58}, \cite{Tati09}. In fact, code-functions can be treated as a generalization of code-words. By explicitly relating code-function distributions and channel input distributions, we convert a mutual information optimization problem over code-function distributions into a residual directed information optimization problem over channel input distributions. Although the theoretical result is in the form of an optimization problem, computing the optimal solution is not feasible. We then turn to investigate computable bounds which are characterized by the causal conditional directed information. Since this new form is a natural generalization of the directed information, the computation is amenable to the dynamic programming approach proposed by Tatikonda and Mitter \cite{Tati09} for the perfect feedback capacity problem.\\
\indent The main contributions of this paper can be summarized as follows: $1)$. We propose a new information theoretic concept, the residual directed information, to identify and capture the effective information flow in communication channels with noisy feedback and then analyze the information flow in the forward channel. $2)$. We prove that, for DMC with typical noisy feedback, no capacity-achieving closed-loop encoding strategy exists under certain reasonable conditions. $3)$. We show a general noisy feedback channel coding theorem in terms of the residual directed information. $4)$. We propose computable bounds on the noisy feedback capacity, which are characterized by the causal conditional directed information.\\
\indent Throughout the paper, capital letters $X,Y,Z,\cdots$ will represent random variables and lower case letters $x,y,z,\cdots$ will represent particular realizations. We use $x^n$ to represent the sequence $(x_1,x_2,\cdots,x_n)$ and $x^0=\emptyset$. We use $\log$ to represent logarithm base $2$.
\section{Technical Preliminaries}
\indent In this section, we review and give some important definitions of probability theory and information theory, which are used throughout the paper. We begin with the following assumption.
\begin{assumption}
Every random variable considered throughout the paper is in a finite set (i.e. $\mathcal{X},\mathcal{Y},\mathcal{Z},\cdots$) with the power set $\sigma$-algebra.
\end{assumption}

\indent Although we restrict our exposition to finite alphabets, most of the results in this paper can be extended to the case of any abstract set (i.e. countably infinite or continuous alphabets ).

\begin{definition}\cite{cover06} (\textit{Entropy})
The entropy $H(X)$ of a discrete random variable $X$ is defined by
\begin{equation*}
H(X)=-\sum_{x\in\mathcal{X}} p(x)\log p(x)
\end{equation*}
\end{definition}

\indent We have the following properties of entropy.\\
\indent (P1) $H(X)\geq 0$.\\
\indent (P2) $H(X,Y)=H(X)+H(Y|X)$.\\
\indent (P3) $H(X)\leq \log\lvert\mathcal{X}\rvert$, where $\lvert\mathcal{X}\rvert$ denotes the cardinality of the finite set $\mathcal{X}$, with equality if and only if $X$ has a uniform distribution.\\

\begin{definition}(\textit{Mutual Information and Its Density})
Consider two random variables $X$ and $Y$ with a joint probability mass function $p(X,Y)$ and marginal probability mass functions $p(X)$
and $p(Y)$. The mutual information $I(X;Y)$ is defined by
\begin{equation*}
I(X;Y)=\mathbb{E}_{p(X,Y)}\log\frac{p(X,Y)}{p(X)p(Y)}
\end{equation*}
and the mutual information density is defined by
\begin{equation*}
i(X;Y)=\log\frac{p(X,Y)}{p(X)p(Y)}
\end{equation*}
\end{definition}

\indent We present three properties of mutual information which will be used later.\\
\indent (P4) $I(X;Y)=H(X)-H(X|Y)$.\\
\indent (P5) $I((X,Y);Z|U)=I(Y;Z|U)+I(X;Z|(Y,U))$.\\
\indent (P6) $I(X;Y|Z)=H(Y|Z)-H(Y|X,Z)$\\
\indent (P$4$) shows the relationship between mutual information and entropy. (P$5$) is Kolmogorov's formula\cite{pinsker64}.\\
\indent Now, we introduce a notion of causal conditional probability with respect to a time ordering of random variables.
\begin{definition}(\textit{Causal Conditional Probability})
Given a time ordering of random variables $(X^n,Y^n)$
\begin{equation}
X_1, Y_1, X_2, Y_2,\cdots, X_n, Y_n
\label{equ2_1}
\end{equation}
where $X^n\in\mathcal{X}^n$ and $Y^n\in\mathcal{Y}^n$, the causal conditional probability is defined by the following expression
\begin{equation*}
\overrightarrow{p}(x^n|y^n)=\prod_{i=1}^n p(x_i|x^{i-1},y^{i-1}) \quad \text{and} \quad \overrightarrow{p}(y^n|x^n)=\prod_{i=1}^n p(y_i|x^{i},y^{i-1}).
\end{equation*}

\indent Next, we present the definition of \textit{directed information} with respect to the time ordering sequence (\ref{equ2_1}).
\begin{definition}(\textit{Directed Information and Its Density})
Given a time ordering of random variables $(X^n,Y^n)$ as (\ref{equ2_1}) where $X^n\in\mathcal{X}^n$ and $Y^n\in\mathcal{Y}^n$, the directed information from a sequence $X^n$ to a sequence $Y^n$ is defined by
\begin{equation*}
I(X^n\rightarrow Y^n)=\mathbb{E}_{p(X^n,Y^n)}\log\frac{\overrightarrow{p}(Y^n|X^n)}{p(Y^n)}
\end{equation*}
and the directed information density is defined by
\begin{equation*}
i(X^n\rightarrow Y^n)=\log\frac{\overrightarrow{p}(Y^n|X^n)}{p(Y^n)}
\end{equation*}
\end{definition}

\indent Note that Massey's definition of directed information \cite{Massey1990} can be easily recovered by the above definition.
\begin{equation*}
\begin{split}
I(X^n\rightarrow Y^n)=&\mathbb{E}_{p(X^n,Y^n)}\log\frac{\overrightarrow{p}(Y^n|X^n)}{p(Y^n)}\\
=&\sum_{x^n\in\mathcal{X}^n,y^n\in\mathcal{Y}^n}p(x^n,y^n)\log\frac{\overrightarrow{p}(y^n|x^n)}{p(y^n)}\\
=&\sum_{x^n\in\mathcal{X}^n,y^n\in\mathcal{Y}^n}p(x^n,y^n)\sum_{i=1}^n\log\frac{p(y_i|x^i,y^{i-1})}{p(y_i|y^{i-1})}\\
=&\sum_{i=1}^n\sum_{x^n\in\mathcal{X}^n,y^n\in\mathcal{Y}^n}p(x^n,y^n)\log\frac{p(y_i|x^i,y^{i-1})}{p(y_i|y^{i-1})}\\
=&\sum_{i=1}^n\sum_{x^i\in\mathcal{X}^i,y^i\in\mathcal{Y}^i}p(x^i,y^i)\log\frac{p(y_i|x^i,y^{i-1})}{p(y_i|y^{i-1})}\\
=&\sum_{i=1}^n I(X^i;Y_i|Y^{i-1})\\
\end{split}
\end{equation*}

\indent We refer the interested readers to \cite{Tati09} for the definition of directed information for an arbitrary time ordering of random variables. Next, we extend the definition of directed information to the causal conditional directed information as follows.
\begin{definition}(\textit{Causal Conditional Directed Information and Its Density})
Given a time ordering of random variables $(X^n,Y^n,Z^n)$
\begin{equation}
X_1, Y_1, Z_1, X_2, Y_2, Z_2, \cdots, X_n, Y_n, Z_n
\end{equation}
where $X^n\in\mathcal{X}^n$, $Y^n\in\mathcal{Y}^n$ and $Z^n\in\mathcal{Z}^n$, the directed information from a sequence $X^n$ to a sequence $Y^n$ causally conditioning on $Z^n$ is defined by
\begin{equation*}
I(X^n\rightarrow Y^n||Z^n)=\mathbb{E}_{p(X^n,Y^n,Z^n)}\log\frac{\overrightarrow{p}(Y^n|X^n,Z^n)}{\overrightarrow{p}(Y^n|Z^n)}
\end{equation*}
and the causal conditional directed information density is defined by
\begin{equation*}
i(X^n\rightarrow Y^n||Z^n)=\log\frac{\overrightarrow{p}(Y^n|X^n,Z^n)}{\overrightarrow{p}(Y^n|Z^n)}
\end{equation*}
\label{def6}
\end{definition}
where
\begin{equation*}
\overrightarrow{p}(y^n|x^n,z^n)=\prod_{i=1}^{n}p(y_i|x^{i},y^{i-1},z^{i-1})
\end{equation*}
\end{definition}

\indent It is easy to verify that
\begin{equation*}
I(X^n\rightarrow Y^n||Z^n)=\sum_{i=1}^{n}I(X^i,Y_i|Y^{i-1},Z^{i-1})
\end{equation*}

\begin{remark}
If Markov chains $Z_{i}^n-Y^{i-1}-Y^i$, $Z_{i}^n-(X^i,Y^{i-1})-Y^i$ hold, we may obtain
\begin{equation*}
\begin{split}
&I(X^n\rightarrow Y^n|Z^n)\\
=&\sum_{i=1}^n I(X^i,Y_i|Y^{i-1},Z^n)\\
=&\sum_{i=1}^n H(Y_i|Y^{i-1},Z^n)-H(Y_i|X^i,Y^{i-1},Z^n)\\
=&\sum_{i=1}^n H(Y_i|Y^{i-1},Z^{i-1})-H(Y_i|X^i,Y^{i-1},Z^{i-1})\\
=&\sum_{i=1}^n I(X^i,Y_i|Y^{i-1},Z^{i-1})\\
=&I(X^n\rightarrow Y^n||Z^n)\\
\end{split}
\end{equation*}
That is, the ``causal conditioning'' and the ``normal conditioning'' coincide.
\label{remark_pre_01}
\end{remark}

\section{Residual Directed Information and Information Flow}
\indent In this section, we first introduce the setup considered in the paper and give a high-level discussion on the failure of using either mutual information or directed information as a measure of the effective information flow through the channel. Then we define a new measure, named \textit{residual directed information}, and derive its properties. Finally, we analyze the information flow in the noisy feedback channel.

\begin{figure}
\begin{center}
\includegraphics[scale=0.70]{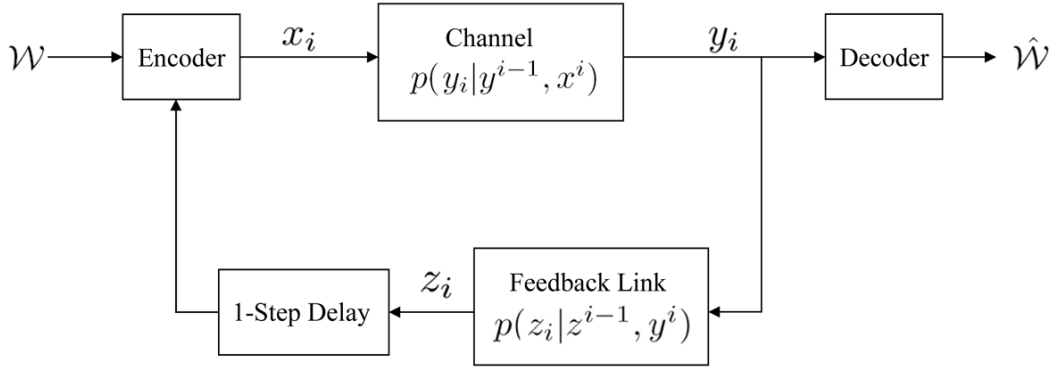}
\caption{Channels with noisy feedback}
\label{figure1}
\end{center}
\end{figure}

\subsection{Noisy Feedback and Causality}
\indent According to Fig.\ref{figure1}, we model the channel at time $i$ as $p(y_i|x^i,y^{i-1})$. The channel output (without any encoding) is fed back to the encoder through a noisy link, which is modeled as $p(z_i|y^i,z^{i-1})$. At time $i$, the deterministic encoder takes the message $\mathcal{W}$ and the past outputs $Z_1,Z_2,\cdots,Z_{i-1}$ of the feedback link, and then produces a channel input $X_i$. Note that the encoder has access to the output of the feedback link with one time-step delay. At time $n$, the decoder takes all the channel outputs $Y_1,Y_2,\cdots,Y_{n}$ and then produces the decoded message $\hat{\mathcal{W}}$. We present the time ordering of these random variables below.
\begin{equation*}
W,X_1,Y_1,Z_1,X_2,Y_2,Z_2,\cdots,X_{n-1},Y_{n-1},Z_{n-1},X_n,Y_n,\hat{W}
\end{equation*}
\indent Note that all initial conditions (e.g. channel, feedback link, channel input, etc.) are automatically assumed to be known in prior by both the encoder and the decoder. Before entering the more technical part of this paper, it is necessary to give a specific definition of ``noisy feedback''.
\begin{definition}(\textit{Noisy Feedback Link})
The feedback link is noisy if for some time instant $i$ there exists no function $g_i$ such that
\begin{equation}
g_i(X^{i},Z^{i},W)= Y^{i}.
\label{eqI_01}
\end{equation}
The feedback link is noiseless if it is not noisy.
\label{def_noisyfb}
\end{definition}
\begin{remark}
This definition states that, for noisy feedback links, not all the channel outputs can be exactly recovered at the encoder side and, therefore, the encoder and decoder lose mutual understanding. In other words, at time instant $i+1$, the encoder cannot access to the past channel outputs $Y^i$ through information $(X^{i},Z^{i},W)$ to produce channel input $X_{i+1}$. We refer ``perfect (ideal) feedback'' to be the case of $Z^i=Y^i$ for all time instant $i$.  Essentially, noiseless feedback is equivalent to perfect feedback since, in both cases, the encoder can access to the channel outputs without any error.
\end{remark}
\begin{example}
Consider the feedback link as $Z_i=Y_i+V_i$ where $V_i$ denotes additive noise at time instant $i$. If channel outputs $Y_i$ only takes value in a set of integers (i.e. $\pm 1, \pm 2, \cdots$) and $V_i$ only takes value in $\lbrace \pm 0.2, \pm 0.4 \rbrace$, then obviously the channel outputs can be exactly recovered at the encoder side. Thus, this feedback link is noiseless even though it is imperfect.
\label{exp01}
\end{example}

\indent Next, we give a definition of \textit{typical noisy feedback link} which will be studied in the next section.
\begin{definition}(\textit{Typical Noisy Feedback Link})
Given channel $\lbrace p(y_i|x^i,y^{i-1})\rbrace_{i=1}^\infty$, the noisy feedback link $\lbrace p(z_i|y^i,z^{i-1})\rbrace_{i=1}^\infty$ is typical if it satisfies
\begin{equation}
\liminf_{n\rightarrow \infty}\frac{1}{n}\sum_{i=1}^{n}H(Z^{i-1}|Y^{i-1})>0
\end{equation}
for any channel input distribution $\lbrace p(x_i|x^{i-1},z^{i-1})\rbrace_{i=1}^\infty$. The noisy feedback link is non-typical if it is not typical.
\label{def_typcialnoise}
\end{definition}
\begin{remark}
This definition implies that the noise in the feedback link must be active consistently over time (e.g. not physically vanishing). In practice, the typical noisy feedback link is the most interesting case for study.
\end{remark}

\begin{example}
Consider a binary symmetric feedback link modeled as $Z_i=Y_i\oplus V_i$ where noise $V_i$ is i.i.d and takes value from $\lbrace 0,1\rbrace$ with equal probability. Then we have
\begin{equation*}
\begin{split}
\liminf_{n\rightarrow \infty}\frac{1}{n}\sum_{i=1}^{n}H(Z^{i-1}|Y^{i-1})=&\liminf_{n\rightarrow \infty}\frac{1}{n}\sum_{i=1}^{n}H(V^{i-1}|Y^{i-1})\\
\geq &\liminf_{n\rightarrow \infty}\frac{1}{n}\sum_{i=1}^{n}H(V_{i-1}|Y^{i-1})\\
\stackrel{(a)}{=}&\liminf_{n\rightarrow \infty}\frac{1}{n}\sum_{i=1}^{n}H(V_{i-1})\\
=&1\\
\end{split}
\end{equation*}
where (a) follows the fact that $Y^{i-1}$ is independent from $V_{i-1}$ due to one step delay. Therefore, this noisy feedback link is typical.
\end{example}

\indent We summarize the family of the feedback link in Fig.\ref{111}.\footnote{In the sequel, the term ``noisy feedback'' refers to ``typical noisy feedback'' unless specified.} We next define the achievable rate and capacity for channels with noisy feedback.
\begin{definition}(\textit{Channel Code})
Consider a message $\mathcal{W}$ which is drawn from an index set $\lbrace 1,2,\cdots,M\rbrace$ and a noisy feedback communication channel $(\mathcal{X}^n, \lbrace p(y_i|x^i,y^{i-1})\rbrace_{i=1}^{n}, \mathcal{Y}^n,\lbrace p(z_i|y^i,z^{i-1})\rbrace_{i=1}^{n},\mathcal{Z}^n)$ with the interpretation that $X_i$ is the input and $Y_i$ is the output/input of the channel/feedback and $Z_i$ is the output of the noisy feedback link at time instant $i$ ($1\leq i\leq n$). Then a $(M,n)$ channel code consists of an index set $\lbrace 1,2,\cdots,M\rbrace$, an encoding function: $\lbrace 1,2,\cdots,M\rbrace\times \mathcal{Z}^{n-1}\rightarrow\mathcal{X}^n$, and a decoding function:$\mathcal{Y}^n\rightarrow \lbrace 1,2,\cdots,M\rbrace$ where the decoding function is a deterministic rule that assigns a guess to each possible received vector.
\end{definition}

\begin{definition}(\textit{Achievable Rate})
The rate $R$ of a $(M,n)$ code is
\begin{equation*}
R=\frac{\log M}{n} \qquad \text{bits per channel use}
\end{equation*}
The rate is said to be achievable if there exists a sequence of $(2^{nR},n)$ codes\footnote{With a slight abuse of notation, we write $nR$ instead of $\lfloor nR\rfloor$ for convenience.} such that the maximal probability of error tends to zero as $n\rightarrow \infty$.
\end{definition}

\begin{definition}(\textit{Channel Capacity})
The capacity of a channel with noisy feedback is the supremum of all achievable rates.
\end{definition}

\indent When there is no feedback from the channel output to the encoder, the maximum of mutual information (i.e. $\max_{p(x^n)}{I(X^n;Y^n)}$) characterizes the maximum information flow through the channel with arbitrarily small probability of decoding error. This quantity is defined as the capacity of the channel. When there is a noiseless feedback, supremizing directed information $I(X^n\rightarrow Y^n)$ over $\overrightarrow{p}(x^n|y^n)$ gives us the feedback capacity \cite{Tati09}, \cite{Kim08_capacity_fb}, \cite{Permuter09}. When there is a noisy feedback, the appropriate measure/characterization of the effective information flow through the channel has been unknown until now. In the next section, we provide the missing measure.
\begin{figure}
\begin{center}
\includegraphics[scale=0.50]{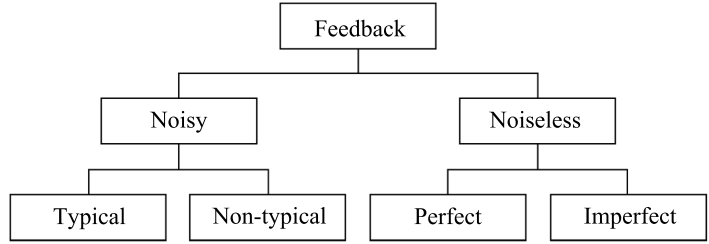}
\caption{Family of Feedback links in Communication systems. The ``typical noisy feedback'' is the case which we are interested in.}
\label{111}
\end{center}
\end{figure}
\subsection{Residual Directed Information}
\indent Based on the ``(causal conditional) directed information'', the \textit{residual directed information} and its density with respect to message $W$ is defined as follows.
\begin{definition}(\textit{Residual Directed Information and Its Density})
\begin{equation}
I^R(X^n(W)\rightarrow Y^n)=I(X^n\rightarrow Y^n)-I(X^n\rightarrow Y^n||W).
\end{equation}
Equivalently,
\begin{equation}
I^R(X^n(W)\rightarrow Y^n)=I(X^n\rightarrow Y^n)-I(X^n\rightarrow Y^n|W).
\label{RDI}
\end{equation}
The residual directed information density is defined as
\begin{equation*}
i^R(X^n(W)\rightarrow Y^n)=i(X^n\rightarrow Y^n)-i(X^n\rightarrow Y^n|W)
\end{equation*}
\end{definition}

\indent The following theorem shows that the residual directed information captures the mutual information between the message and the channel outputs which we refer to the \textit{effective information flow}.
\begin{theorem}
If $X^n$ and $Y^n$ are the inputs and outputs, respectively, of a discrete channel with noisy feedback, as shown in Fig.\ref{figure1}, then
\begin{equation*}
I(W;Y^n)=I^R(X^n(W)\rightarrow Y^n)=I(X^n\rightarrow Y^n)-I(X^n\rightarrow Y^n|W).
\end{equation*}
\label{thm3_1}
\end{theorem}
\begin{proof}
\begin{equation*}
\begin{split}
&I(W;Y^n)\\
=&H(Y^n)-H(Y^n|W)\\
=&\sum_{i=1}^{n}H(Y_i|Y^{i-1})-\sum_{i=1}^{n}H(Y_i|Y^{i-1},W)\\
=&\sum_{i=1}^{n}H(Y_i|Y^{i-1})-\sum_{i=1}^{n}H(Y_i|Y^{i-1},W,X^i)-(\sum_{i=1}^{n}H(Y_i|Y^{i-1},W)-\sum_{i=1}^{n}H(Y_i|Y^{i-1},W,X^i))\\
\stackrel{(a)}{=}&\sum_{i=1}^{n}H(Y_i|Y^{i-1})-\sum_{i=1}^{n}H(Y_i|Y^{i-1},X^i)-(\sum_{i=1}^{n}H(Y_i|Y^{i-1},W)-\sum_{i=1}^{n}H(Y_i|Y^{i-1},W,X^i)) \\
=&\sum_{i=1}^{n}I(X^i;Y_i|Y^{i-1})-\sum_{i=1}^{n}I(X^i;Y_i|Y^{i-1},W)\\
=&I(X^n\rightarrow Y^n)-I(X^n\rightarrow Y^n|W) \\
\stackrel{(b)}{=}&I^R(X^n(W)\rightarrow Y^n) \\
\end{split}
\end{equation*}
where (a) follows from the Markov chain $W - (X^i,Y^{i-1})- Y_i$. Line (b) follows from the definition of residual directed information.\\
\end{proof}

\begin{remark} This theorem implies that, for noisy feedback channels, the directed information $I(X^n\rightarrow Y^n)$ captures both the effective information flow (i.e. $I(W;Y^n)$) generated by the message and the redundant information flow (i.e. $I(X^n\rightarrow Y^n|W)$) generated by the \textit{feedback noise} (dummy message). Since only $I(W;Y^n)$ is the relevant quantity for channel capacity, the well-known directed information clearly fails to characterize the noisy feedback capacity.
\end{remark}

\indent In the following corollary, we explore some properties of the residual directed information.
\begin{corollary}
The residual directed information $I^R(X^n(W)\rightarrow Y^n)$ satisfies the following properties:
\begin{enumerate}
\item $I^R(X^n(W)\rightarrow Y^n)\geq 0$  (with equality if and only if the message set $W$ and channel outputs $Y^n$ are independent.)
\item $I^R(X^n(W)\rightarrow Y^n)\leq I(X^n\rightarrow Y^n)\leq I(X^n;Y^n)$.
\end{enumerate}
 The first equality holds if the feedback is perfect. The second equality holds if there is no feedback.
\label{col3_1}
\end{corollary}
\begin{proof}
\indent 1). Follows from Theorem \ref{thm3_1}, $I^R(X^n(W)\rightarrow Y^n)=I(W;Y^n)\geq 0$. The necessary and sufficient condition of $I^R(X^n(W)\rightarrow Y^n)=0$ is obvious by looking at $I(W;Y^n)$.\\
\indent 2). Since $I(X^n\rightarrow Y^n|W)=\sum_{i=1}^n I(X^i;Y_i|Y^{i-1},W)\geq 0$ (equality holds for the perfect feedback case),
\begin{equation*}
\begin{split}
I^R(X^n(W)\rightarrow Y^n)&=I(X^n\rightarrow Y^n)-I(X^n\rightarrow Y^n|W)\leq I(X^n\rightarrow Y^n)\\
\end{split}
\end{equation*}
\indent The proof of the second inequality $I(X^n\rightarrow Y^n)\leq I(X^n;Y^n)$ is presented in \cite{Massey1990}.
\end{proof}

\begin{figure}
\begin{center}
\includegraphics[scale=0.60]{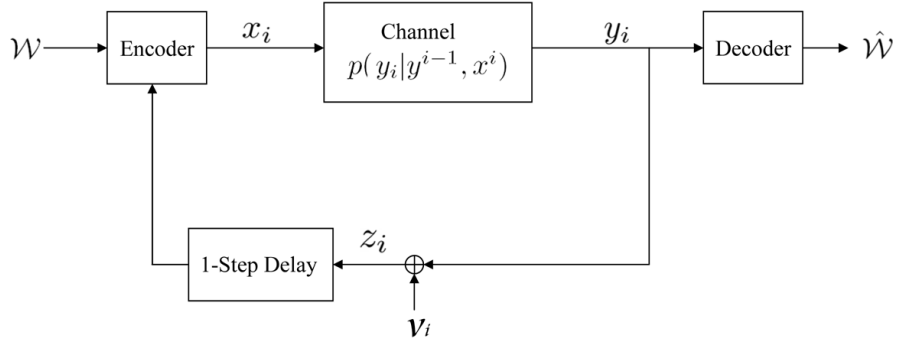}
\caption{Channels with additive noise feedback}
\label{figure5}
\end{center}
\end{figure}

\subsection{Information Flow in Noisy Feedback Channels}
\indent To gain more insight in the information flow of noisy feedback channels, we apply the new concept to channels with additive noise feedback and analyze its information flow. See Fig.\ref{figure5}. We present the time ordering of these random variables below\footnote{$Z_i$ is not shown in the time ordering since we have $Z_i=Y_i+V_i$.}.
\begin{equation*}
W,X_1,Y_1,V_1,X_2,Y_2,V_2,\cdots,X_{n-1},Y_{n-1},V_{n-1},X_n,Y_n,\hat{W}
\end{equation*}

\begin{corollary}
If $X^n$ and $Y^n$ are the inputs and outputs, respectively, of a discrete channel with additive noise feedback, as shown in Fig.\ref{figure5}, then
\begin{equation*}
I(X^n\rightarrow Y^n)=I(W;Y^n)+I(V^{n-1};Y^n)+I(W;V^{n-1}|Y^n)
\end{equation*}
\label{col3_3}
\end{corollary}

\begin{proof}
We herein adopt a derivation methodology similar to the one in Theorem \ref{thm3_1}.
\begin{equation*}
\begin{split}
&I(W;Y^n)\\
=&H(Y^n)-H(Y^n|W)\\
=&\sum_{i=1}^{n}H(Y_i|Y^{i-1})-\sum_{i=1}^{n}H(Y_i|Y^{i-1},W)\\
=&\sum_{i=1}^{n}H(Y_i|Y^{i-1})-\sum_{i=1}^{n}H(Y_i|Y^{i-1},W,V^{i-1})-(\sum_{i=1}^{n}H(Y_i|Y^{i-1},W)-\sum_{i=1}^{n}H(Y_i|Y^{i-1},W,V^{i-1}))\\
\stackrel{(a)}{=}&\sum_{i=1}^{n}H(Y_i|Y^{i-1})-\sum_{i=1}^{n}H(Y_i|Y^{i-1},W,Z^{i-1})-(\sum_{i=1}^{n}H(Y_i|Y^{i-1},W)-\sum_{i=1}^{n}H(Y_i|Y^{i-1},W,V^{i-1}))\\
=&\sum_{i=1}^{n}H(Y_i|Y^{i-1})-\sum_{i=1}^{n}H(Y_i|Y^{i-1},X^{i})-(\sum_{i=1}^{n}H(Y_i|Y^{i-1},W)-\sum_{i=1}^{n}H(Y_i|Y^{i-1},W,V^{i-1}))\\
=&\sum_{i=1}^{n}I(X^i;Y_i|Y^{i-1})-\sum_{i=1}^{n}I(V^{i-1};Y_i|Y^{i-1},W)\\
=&I(X^n\rightarrow Y^n)-I(V^{n-1}\rightarrow Y^n|W)
\end{split}
\end{equation*}
where (a) follows from the fact that $Z^{i-1}=Y^{i-1}+V^{i-1}$. Next,
\begin{equation*}
\begin{split}
I(V^{n-1}\rightarrow Y^n|W)\stackrel{(b)}{=}&I(V^{n-1};Y^n|W) \\
=&H(V^{n-1}|W)-H(V^{n-1}|Y^n,W)\\
\stackrel{(c)}{=}&H(V^{n-1})-H(V^{n-1}|Y^n)+H(V^{n-1}|Y^n)-H(V^{n-1}|Y^n,W)\\
=&I(V^{n-1};Y^n)+I(W;V^{n-1}|Y^n)\\
\end{split}
\end{equation*}
where (b) follows from the fact that there exists no feedback from $Y^n$ to $V^{n-1}$ and (c) follows from the fact that the noise $V^{n-1}$ is independent from $W$. Putting previous equations together, the proof is complete.
\end{proof}

\begin{figure}
\begin{center}
\includegraphics[scale=0.55]{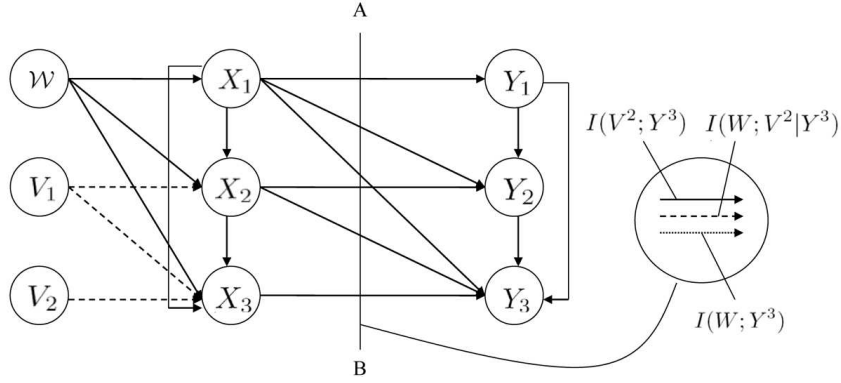}
\caption{The information flow of channels with additive noise feedback}
\label{figure2}
\end{center}
\end{figure}

\indent Corollary \ref{col3_3} allows us to explicitly interpret the information flow on a dependency graph (e.g. $N=3$). See Fig.\ref{figure2}. The solid lines from message $\mathcal{W}$ to sequence $X^3$ represent the dependence of $X^3$ on $W$. The dotted lines from additive noise $V^2$ to sequence $X^3$ represent the dependence of $X^3$ on $V^2$. The dependence of the channel inputs $X^3$ on the channel outputs $Y^2$ is not shown in the graph since the directed information only captures the information flow from $X^3$ to $Y^3$ \cite{Massey1990}. As it is shown in the zoomed circle, the directed information flow from $X^3$ to $Y^3$ (through cut $A-B$) implicitly contains three sub-information flows wherein the mutual information $I(W;Y^3)$ and $I(V^2;Y^3)$ measure the message-transmitting and the noise-transmitting information flows, respectively. The feedback noise $V^2$ is treated as a dummy message which also needs to be recovered by the decoder. The conditional mutual information $I(W;V^2|Y^3)$ quantifies the mixed information flow between the message-transmitting and noise-transmitting flows. Essentially, the second term in the residual directed information (i.e. $I(X^n\rightarrow Y^n|W)$) precisely captures the non-message transmitting information flows (i.e. $I(V^{n-1};Y^n)$ and $I(W;V^{n-1}|Y^n)$). Therefore, the residual directed information should be a proper measure to work with for channels with noisy feedback.\\
\indent Understanding the information flow in noisy feedback channels leads us to a higher level to investigate the noisy feedback problem and performs as the basis to develop fruitful results (to be seen later).

\section{Discrete Memoryless Channel With Noisy Feedback}
\label{sec_DMC_nfb}
\indent With the new concept and the picture of the information flow in hand, we now concentrate on DMC with noisy feedback. We show a negative yet fundamental result that the capacity is not achievable by using any non-trivial closed-loop encoder. In other words, exploiting the information from the feedback link is actually detrimental to achieving the maximal achievable rate. We first give some necessary definitions below.

\subsection{Discrete Memoryless Channel and Typical Closed-Loop Encoder}
\begin{definition}(\textit{Discrete Memoryless Channel})
A discrete memoryless channel is a discrete channel satisfying
\begin{equation*}
p(y_i|x^i,y^{i-1})=p(y_i|x_i)
\end{equation*}
\end{definition}

\begin{definition}(\textit{Typical Closed-Loop Encoder })
Given a channel $\lbrace p(y_i|x^i,y^{i-1})\rbrace_{i=1}^\infty$, a noisy feedback link $\lbrace p(z_i|y^i,z^{i-1})\rbrace_{i=1}^\infty$, an encoder is defined as a typical closed-loop encoder if it satisfies
\begin{equation*}
\liminf_{n\rightarrow \infty}\frac{1}{n}I(Z^{n-1}\rightarrow Y^n)>0.
\end{equation*}
For the additive noise feedback case as shown in Fig.\ref{figure5}, the condition is equivalent to
\begin{equation*}
\liminf_{n\rightarrow \infty}\frac{1}{n}I(V^{n-1}; Y^n)>0.
\end{equation*}
\label{def_typicalencoder}
\end{definition}
\begin{remark}
The equivalence is straightforward to check. That is,
\begin{equation*}
\begin{split}
\liminf_{n\rightarrow \infty}\frac{1}{n}I(Z^{n-1}\rightarrow Y^n)=&\liminf_{n\rightarrow \infty}\frac{1}{n}\sum_{i=1}^{n}H(Y_i|Y^i)-H(Y_i|Y^{i-1},Z^{i-1})\\
=&\liminf_{n\rightarrow \infty}\frac{1}{n}\sum_{i=1}^{n}H(Y_i|Y^i)-H(Y_i|Y^{i-1},V^{i-1})\\
=&\liminf_{n\rightarrow \infty}\frac{1}{n}I(V^{n-1}\rightarrow Y^n)\\
\stackrel{(a)}{=}&\liminf_{n\rightarrow \infty}\frac{1}{n}I(V^{n-1};Y^n).\\
\end{split}
\end{equation*}
where (a) follows the fact that there is no feedback from $Y$ to $V$ and thus the mutual information and the directed information coincide.
\end{remark}

\begin{remark}
This definition implies that a typical closed-loop encoder should non-trivially take feedback information $Z^{n-1}$ to produce channel inputs $X^{n}$ over time. It is easy to verify that an encoder is non-typical if it discards all feedback information (i.e. open-loop encoder) or only extracts feedback information for finite time instants.
\end{remark}

\begin{remark}
The typical closed-loop encoder is only well-defined under the assumption of typical noisy feedback (definition \ref{def_typcialnoise}). Otherwise, for any encoder, we have
\begin{equation*}
\begin{split}
\liminf_{n\rightarrow \infty}\frac{1}{n}I(Z^{n-1}\rightarrow Y^n)=&\liminf_{n\rightarrow \infty}\frac{1}{n}\sum_{i=1}^n I(Z^{i-1}; Y_i|Y^{i-1})\\
=&\liminf_{n\rightarrow \infty}\frac{1}{n}\sum_{i=1}^n H(Z^{i-1}|Y^{i-1})-H(Z^{i-1}|Y^{i})\\
\leq &\liminf_{n\rightarrow \infty}\frac{1}{n}\sum_{i=1}^n H(Z^{i-1}|Y^{i-1})\\
= &0.
\end{split}
\end{equation*}
\end{remark}

\indent Now, we present the main theorem of this section.
\begin{theorem}
The capacity $C_{FB}^{noise}$ of a discrete memoryless channel with noisy feedback equals the non-feedback capacity $C$. The capacity $C_{FB}^{noise}$ is not achievable by implementing any typical closed-loop encoder. Alternatively, any capacity-achieving encoder is non-typcial. Furthermore, the rate-loss by implementing a typical closed-loop encoder is lower bounded by $\liminf_{n\rightarrow \infty}\frac{1}{n}I(Z^{n-1}\rightarrow Y^n)$.\footnote{The ``rate-loss'' refers to the gap between the capacity $C$ and the achievable rate $R$. Given a channel $\lbrace p(y_i|x^i,y^{i-1})\rbrace_{i=1}^\infty$ and a noisy feedback link $\lbrace p(z_i|y^i,z^{i-1})\rbrace_{i=1}^\infty$, the value of $I(Z^{n-1}\rightarrow Y^n)$ only depends on the channel input distributions $\lbrace p(x_i|x^{i-1},z^{i-1})\rbrace_{i=1}^\infty$ induced by the implemented encoder.}
\label{Thm_4_1}
\end{theorem}

\begin{remark}
This negative result implies that it is impossible to find a capacity-achieving feedback coding scheme for DMC with noisy feedback whereas it is possible in perfect feedback case (e.g. Schalkwijk-Kailath scheme). For example, \cite{Martins08} has proposed a linear coding scheme for AWGN channel with bounded feedback noise and \cite{Chance10} has proposed a concatenated coding scheme for AWGN channel with noisy feedback. It is easy to check that both of these closed-loop encoders are typical and therefore both coding schemes cannot achieve the capacity unless, as discussed in \cite{Martins08,Chance10}, the feedback additive noise is shrinking to zero (i.e. non-typical noisy feedback).
\end{remark}

\begin{remark}
Theorem \ref{Thm_4_1} indicates that the noisy feedback is unfavorable in the sense of achievable rate. However, using noisy feedback still provides many benefits as mentioned in the Introduction. Furthermore, from a control theoretic point of view, (noisy) feedback is necessary for stabilizing unstable plants and achieving certain performances. Therefore, we need a tradeoff while using noisy feedback.
\end{remark}

\indent Before moving to prove the main theorem, we need the following lemma.
\begin{lemma}
For any typical closed-loop encoder,
\begin{equation*}
\liminf_{n\rightarrow \infty}\frac{1}{n}I(X^{n}\rightarrow Y^n|W)>0.
\end{equation*}
\label{lem_dmc}
\end{lemma}

\begin{proof}
For any $0\leq i\leq n$, we have
\begin{equation*}
\begin{split}
I(W;Z_i|Y^i,Z^{i-1})=&H(Z_i|Y^i,Z^{i-1})-H(Z_i|Y^{i},Z^{i-1},W)\\
=&H(Z_i|Y^i,Z^{i-1})-H(Z_i|Y^{i},Z^{i-1})\\
=&0.\\
\end{split}
\end{equation*}
Then,
\begin{equation*}
\begin{split}
I(W;(Y^n,Z^{n-1}))=&I(W;(Y^n,Z^n))-I(W;Z_n|Y^n,Z^{n-1})\\
=&\sum_{i=1}^{n}I(W;(Y_i,Z_i)|Y^{i-1},Z^{i-1})\\
=&\sum_{i=1}^{n}I(W;Y_i|Y^{i-1},Z^{i-1})+I(W;Z_i|Y^i,Z^{i-1})\\
=&\sum_{i=1}^{n}H(Y_i|Y^{i-1},Z^{i-1})-H(Y_i|Y^{i-1},Z^{i-1},W)\\
=&\sum_{i=1}^{n}H(Y_i|Y^{i-1},Z^{i-1})-H(Y_i|Y^{i-1},Z^{i-1},W,X^i)\\
=&\sum_{i=1}^{n}H(Y_i|Y^{i-1},Z^{i-1})-H(Y_i|Y^{i-1},X^i)\\
\end{split}
\end{equation*}
We investigate another equality as follows.
\begin{equation*}
\begin{split}
&I(X^{n}\rightarrow Y^n)-I(Z^{n-1}\rightarrow Y^n)\\
=&\sum_{i=1}^{n}H(Y_i|Y^i)-H(Y_i|Y^{i-1},X^{i})-H(Y_i|Y^i)+H(Y_i|Y^{i-1},Z^{i-1})\\
=&\sum_{i=1}^{n}H(Y_i|Y^{i-1},Z^{i-1})-H(Y_i|Y^{i-1},X^{i})\\
\end{split}
\end{equation*}

Combine the above equalities, we have
\begin{equation*}
\begin{split}
I(Z^{n-1}\rightarrow Y^n)=&I(X^{n}\rightarrow Y^n)-I(W;(Y^n,Z^{n-1}))\\
\stackrel{(a)}{=}&I(W;Y^n)+I(X^n\rightarrow Y^n|W)-I(W;(Y^n,Z^{n-1}))\\
=&I(X^n\rightarrow Y^n|W)-I(W;Z^{n-1}|Y^n)\\
\end{split}
\end{equation*}
where $(a)$ follows from Theorem \ref{thm3_1}. According to the definition of typical closed-loop encoder, the proof is complete.
\end{proof}

\indent Now we are ready to prove Theorem \ref{Thm_4_1}.
\begin{proof}
\indent Firstly, we prove that
\begin{equation*}
C_{FB}^{noise}=C=\max_{p(x)}I(X;Y)
\end{equation*}
\indent Since a nonfeedback channel code is a special case of a noisy feedback channel code, any rate that can be achieved without feedback can be achieved with noisy feedback. Therefore, we have $C_{FB}^{noise}\geq C$. Given a noisy feedback link, we clearly have $C_{FB}^{noise}\leq C_{FB}$ where $C_{FB}$ is the capacity of channels with perfect feedback. As $C=C_{FB}$ for DMC \cite{shannon56}, we have $C_{FB}^{noise}=C=\max_{p(x)}I(X;Y)$.\\
\indent Next, we show that for any typical closed-loop encoder, the achievable rates $R$ are strictly less then $C$ and the difference is lower bounded by  $\liminf_{n\rightarrow \infty}\frac{1}{n}I(Z^{n-1}\rightarrow Y^n)$. Let $\mathcal{W}$ be uniformly distributed over $\lbrace 1,2,\cdots,2^{nR} \rbrace$ and $P_e^{(n)}=Pr(W\neq\hat{W})$ with $P_e^{(n)}\rightarrow 0$ as $n\rightarrow \infty$. Then
\begin{equation*}
\begin{split}
nR&=H(W)\\
&=H(W|\hat{W})+I(W;\hat{W})\\
&\stackrel{(a)}{\leq} 1+P_e^{(n)}nR+I(W;\hat{W})\\
&\stackrel{(b)}{\leq} 1+P_e^{(n)}nR+I(W;Y^n) \\
\end{split}
\end{equation*}
where (a) and (b) follow from Fano's inequality and Data-processing inequality, respectively.\\
\indent Next,
\begin{equation*}
\begin{split}
I(W;Y^n)&=I^R(X^n(W)\rightarrow Y^n)\\
&=I(X^n\rightarrow Y^n)-I(X^n\rightarrow Y^n|W)\\
&=\sum_{i=1}^{n}H(Y_i|Y^{i-1})-\sum_{i=1}^{n}H(Y_i|X^i,Y^{i-1})-I(X^n\rightarrow Y^n|W)\\
&\stackrel{(c)}{=}\sum_{i=1}^{n}H(Y_i|Y^{i-1})-\sum_{i=1}^{n}H(Y_i|X_i)-I(X^n\rightarrow Y^n|W)\\
&\stackrel{(d)}{\leq}\sum_{i=1}^{n}H(Y_i)-\sum_{i=1}^{n}H(Y_i|X_i)-I(X^n\rightarrow Y^n|W)\\
&=\sum_{i=1}^{n}I(X_i;Y_i)-I(X^n\rightarrow Y^n|W)\\
&\leq nC-I(X^n\rightarrow Y^n|W)\\
\end{split}
\end{equation*}
where (c) follows from the definition of DMC and (d) follows from the fact that removing conditioning increases entropy.\\
\indent Putting these together, we have
\begin{equation*}
R\leq \frac{1}{n}+P_e^{(n)}R+C-\frac{1}{n}I(X^n\rightarrow Y^n|W)
\end{equation*}
Therefore,
\begin{equation*}
\begin{split}
R&\leq \liminf_{n\rightarrow\infty}\lbrace\frac{1}{n}+P_e^{(n)}R+C-\frac{1}{n}I(X^n\rightarrow Y^n|W)\rbrace\\
&= C-\liminf_{n\rightarrow\infty}\frac{1}{n}I(X^n\rightarrow Y^n|W)\\
\end{split}
\end{equation*}

\indent According to the proof of Lemma \ref{lem_dmc}, we have
\begin{equation*}
\begin{split}
R&\leq  C-\liminf_{n\rightarrow\infty}\frac{1}{n}(I(Z^{n-1}\rightarrow Y^n)+I(W;Z^{n-1}|Y^n))\\
&\leq C-\liminf_{n\rightarrow\infty}\frac{1}{n}I(Z^{n-1}\rightarrow Y^n)\\
\end{split}
\end{equation*}
\indent By the definition of typical closed-loop encoder, the proof is complete.
\end{proof}

\begin{figure}
\begin{center}
\includegraphics[scale=0.75]{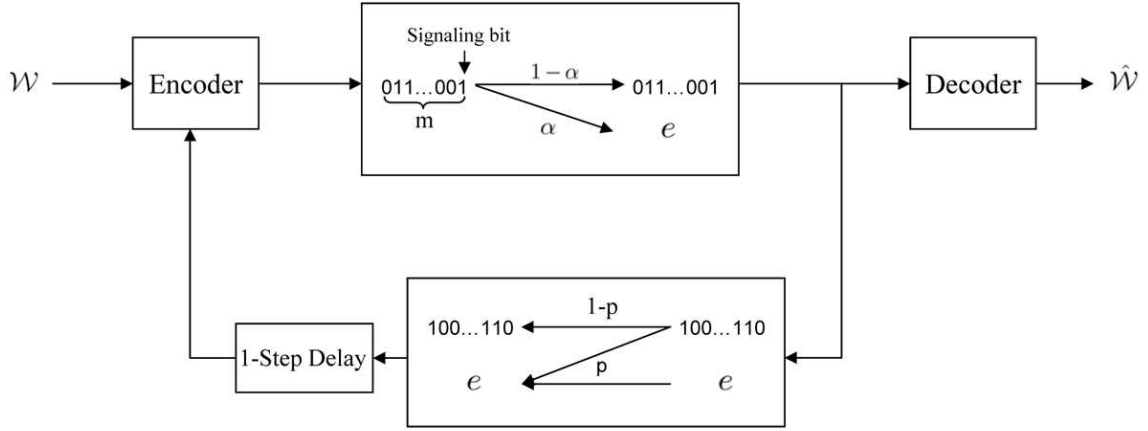}
\caption{Binary codeword erasure channel/feedback}
\label{figure7}
\end{center}
\end{figure}
\subsection{Example}
\indent We give an example of communication through DMC with typical noisy feedback, from which we may get insight on how feedback ``noise'' reduces effective transmission rate and how signaling helps rebuild the coordination between the transmitter and the receiver. Consider a binary codeword erasure channel (BCEC) with a noisy feedback as shown in Fig.\ref{figure7}. The channel input is a m-bit codeword. This input codeword will be reliably transmitted with probability $1-\alpha$, and otherwise get erased with probability $\alpha$. Similarly, we assume a noisy feedback link with erasure probability $p$. It is obvious that the capacity of this channel is $C_{FB}^{noise}=m(1-\alpha)$. One simple but nonoptimum encoding strategy is the following: use the first bit in every m-bit codeword as a signaling bit (i.e. $1$ refers to a retransmitted m-bit codeword while $0$ refers to a new one). If the output of the feedback link is $e$, the encoder will retransmit the previous codeword with signaling ``1'', otherwise, transmit the next codeword with signaling ``0''. Under this strategy, the decoder can recover the message with arbitrarily small error due to the signaling bit. Next, we analyze the achievable rate of this strategy. Assume that $n$ bits information need to be transmitted and $n$ is sufficient large. Then $\alpha n$ bits will be lost and $(1-\alpha)n$ bits will reliably get through. Due to the noisy feedback, the encoder will retransmit $b_1=\alpha n+p(1-\alpha)n$ bits. Similarly, $\alpha b_1$ bits will be lost and $(1-\alpha)b_1$ bits will get through. Then the encoder will retransmit $b_2=\alpha b_1+p(1-\alpha)b_1$ bits. After retransmit $t$ times with $t\rightarrow \infty$, the achievable transmission rate $R$ is
\begin{equation*}
\begin{split}
R&=\lim_{t\rightarrow \infty} \frac{\log{2^{n}}}{\frac{1}{m-1}(n+b_1+b_2+\cdots+b_t)}\\
&=\lim_{t\rightarrow \infty} \frac{n(m-1)}{n+\frac{(\alpha n+p(1-\alpha)n)(1-\alpha+p(1-\alpha))^t}{1-\alpha+p(1-\alpha))}}\\
&=\frac{n(m-1)}{n+\frac{(\alpha n+p(1-\alpha)n)}{1-(\alpha+p(1-\alpha))}}\\
&=(m-1)(1-p)(1-\alpha)\\
\end{split}
\end{equation*}
Then we have
\begin{equation*}
\frac{R}{C_{FB}^{noise}}=(1-p)(1-\frac{1}{m}).
\end{equation*}

\indent Here, it shows that the loss of transmission rate is caused by two factors: uncertainty in the feedback link and signaling in the forward channel. If $p=0$ (i.e. perfect feedback) and $m\rightarrow \infty$ (i.e. the signaling bit could be neglected), we have $R=C_{FB}^{noise}$. Additionally, we should notice an interesting fact in this example that the loss of effective transmission rate is independent of the noise in the forward channel.

\section{A Channel Coding Theorem and Computable Bounds on the Capacity}
\label{sec_channelcoding}
\begin{figure}
\begin{center}
\includegraphics[scale=0.70]{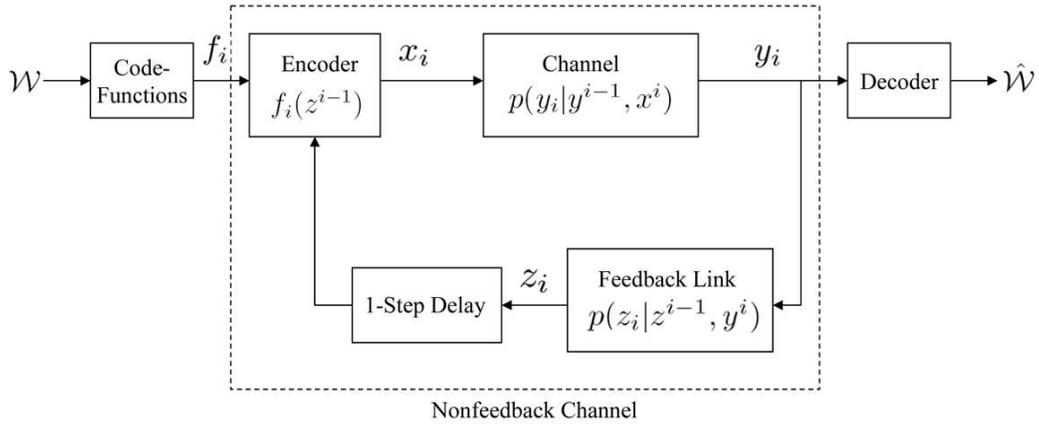}
\caption{Channels with noisy feedback (a code-function representation)}
\label{figure4}
\end{center}
\end{figure}
\indent In this section, we first show that the residual directed information can be used to characterize the capacity of finite alphabet channels with noisy feedback. As we will discuss, this characterization has nice features and provides much insight in the noisy feedback capacity. However, the computation of this characterization is in general intractable. We then propose computable bounds which are characterized by the causal conditional directed information.\\
\indent We first formulate the channel coding problem. Here, we require the use of code-functions as opposed to codewords, as shown in Fig.\ref{figure4}. Briefly, at time $0$, we choose a message from a message set $\mathcal{W}$. This message is associated with a sequence of code-functions. Then from time $1$ to $n$, we use the channels to transmit information sequentially based on the corresponding code-function. At time $n+1$, we decode the message as $\hat{\mathcal{W}}$. We now give a formal definition of this communication scheme, which extends the description presented in \cite{Tati09}.
\begin{definition}(\textit{Communication Scheme for Channels with Noisy Feedback: A Code-function Representation})\\
\indent $1$. A message set is a set $\mathcal{W}\in\lbrace1,2,\cdots,M\rbrace$\\
\indent $2$. A channel code-function is a sequence of $n$ deterministic measurable maps $f^n=\lbrace f_i\rbrace_{i=1}^n$ ($f\in \mathcal{F}$) such that $f_i: \mathcal{Z}^{i-1} \rightarrow \mathcal{X}$ which takes $z^{i-1}\mapsto x_i$. \\
\indent $3$. A channel encoder is a set of $M$ channel code-functions, denoted by $\lbrace f^n[w]\rbrace_{w=1}^M$.\\
\indent $4$. A channel is a family of conditional probability $\lbrace p(y_i|x^i,y^{i-1})\rbrace_{i=1}^n$.\\
\indent $5$. A noisy feedback link is a family of conditional probability $\lbrace p(z_i|y^i,z^{i-1})\rbrace_{i=1}^n$.\\
\indent $6$. A channel decoder is a map $g:\mathcal{Y}^n \rightarrow \mathcal{W}$ which takes $y^{n}\mapsto w$.\\
\label{def5_01}
\end{definition}

\indent Based on the above communication scheme, we redefine the channel code and $\epsilon$-achievable rate in terms of code-functions.
\begin{definition}(\textit{Channel Code})
A $(n,M,\epsilon)$ channel code over time horizon n consists of $M$ code-functions $\lbrace f^n[w]\rbrace_{w=1}^{M}$, a channel decoder $g$, and an error probability satisfying
\begin{equation*}
\frac{1}{M}\sum_{w=1}^M p(w\neq g(y^n)|w)\leq \epsilon
\end{equation*}
\end{definition}

\begin{definition}(\textit{$\epsilon$-achievable Rate})
$R\geq 0$ is an $\epsilon$-achievable rate if, for every $\epsilon>0$, there exist, for all sufficiently large n, a $(n,M,\epsilon)$ channel code with rate
\begin{equation*}
\frac{\log M}{n}\geq R-\epsilon
\end{equation*}
The maximum $\epsilon$-achievable rate is called the $\epsilon$-capacity, denoted by $C_{FB}^{noise}(\epsilon)$. The channel capacity $C_{FB}^{noise}$ is defined as the maximal rate that is $\epsilon$-achievable for all $0<\epsilon <1$. Clearly, $C_{FB}^{noise}=\lim_{\epsilon\rightarrow 0}C_{FB}^{noise}(\epsilon)$
\end{definition}

\indent The channel coding problem is to search for a sequence of $(n,M,\epsilon)$ channel codes under which the achievable rate is maximized as $n$ goes to $\infty$. In order to construct a general channel coding theorem (i.e. no restrictions on channels and input/output alphabets, such as stationary, ergodic, $\cdots$), we introduce the following two probabilistic limit operations \cite{bookhan03}.
\begin{definition}(\textit{Probabilistic Limit})
The limit superior in probability for any sequence $(X_1,X_2,\cdots)$ is defined by
\begin{equation*}
p-\limsup_{n\rightarrow\infty}X_n=\inf\lbrace \alpha |\lim_{n\rightarrow\infty} Prob\lbrace X_n>\alpha \rbrace =0\rbrace
\end{equation*}
Similarly, the limit inferior in probability for any sequence $(X_1,X_2,\cdots)$ is defined by
\begin{equation*}
p-\liminf_{n\rightarrow\infty}X_n=\sup\lbrace \beta |\lim_{n\rightarrow\infty} Prob\lbrace X_n<\beta \rbrace =0\rbrace
\end{equation*}
\end{definition}

\indent Next, we introduce some notations.
\begin{equation*}
\begin{split}
\underline{I}(X;Y)&=p-\liminf_{n\rightarrow\infty}\frac{1}{n}i(X^n;Y^n)\\
\overline{I}(X;Y)&=p-\limsup_{n\rightarrow\infty}\frac{1}{n}i(X^n;Y^n)\\
\underline{I}^R(X(F)\rightarrow Y)&=p-\liminf_{n\rightarrow\infty}\frac{1}{n}i^R(X^n(F^n)\rightarrow Y^n)\\
\overline{I}^R(X(F)\rightarrow Y)&=p-\limsup_{n\rightarrow\infty}\frac{1}{n}i^R(X^n(F^n)\rightarrow Y^n)\\
\end{split}
\end{equation*}

\indent As done in \cite{Tati09}, it is convenient to consider the noisy feedback channel problem as a regular nonfeedback problem from the input alphabet $\mathcal{F}$ and output alphabet $\mathcal{Y}$ as shown in Fig.\ref{figure4}. This consideration provides us with an approach to prove the channel coding theorem for channels with noisy feedback. Recall that the capacity of nonfeedabck channels is characterized as follows \cite{Verdu94}.
\begin{theorem}(Non-feedback Channel Capacity)
For any channel with arbitrary input and output alphabets $\mathcal{F}$ and $\mathcal{Y}$, the channel capacity $C$ is given by
\begin{equation*}
C=\sup_{F}\underline{I}(F;Y)
\end{equation*}
\label{thm4_1}
where $\sup_{F}$ denotes the supremum with respect to all the input processes $F$.
\end{theorem}

\indent However, before applying the above result, we need understand the inherent connection between the equivalent nonfeedabck channel and the original channel with noisy feedback link. Moreover, as supremizing the mutual information over code-function $F$ is inconvenient, we need create a connection between the nonfeedback channel input distribution $\lbrace p(f^n)\rbrace$ and the original channel input distribution such that we can still work on the original channel input. These two issues are the main technical steps toward the channel coding theorem. We provide these results as lemmas in the next subsection. Then, we prove the channel coding theorem along the lines of the proof of Theorem \ref{thm4_1}.

\subsection{Technical Lemmas}
\indent We first show an equality of information densities between the nonfeedback channel $\mathcal{F}^n\rightarrow \mathcal{Y}^n$ and the original channel $\mathcal{X}^n\rightarrow \mathcal{Y}^n$.
\begin{lemma}
\begin{equation*}
i(F^n;Y^n)=i^R(X^n(F^n)\rightarrow Y^n)
\end{equation*}
where $i^R(X^n(F^n)\rightarrow Y^n)$ is defined as
\begin{equation*}
i^R(X^n(F^n)\rightarrow Y^n)=i(X^n\rightarrow Y^n)-i(X^n\rightarrow Y^n||F^n).
\end{equation*}
\label{lemma4_3}
\end{lemma}
\begin{proof}
\begin{equation*}
\begin{split}
i(F^n;Y^n)&=\log \frac{p(F^n,Y^n)}{p(F^n)p(Y^n)}\\
&=\log \frac{\prod_{i=1}^{n}p(F_i,Y_i|F^{i-1},Y^{i-1})}{p(F^n)p(Y^n)}\\
&=\log \frac{\prod_{i=1}^{n}p(Y_i|F^{i},Y^{i-1})p(F_i|F^{i-1},Y^{i-1})}{p(F^n)p(Y^n)}\\
&\stackrel{(a)}{=}\log \frac{\prod_{i=1}^{n}p(Y_i|F^{i},Y^{i-1})p(F_i|F^{i-1})}{p(F^n)p(Y^n)}  \\
&=\log \frac{\vec{p}(Y^n|F^n,X^n)}{p(Y^n)}-\log \frac{\vec{p}(Y^n|F^n,X^n)}{\prod_{i=1}^{n}p(Y_i|F^{i},Y^{i-1})}\\
&=\log \frac{\prod_{i=1}^n p(Y_i|F^i,X^i,Y^{i-1})}{p(Y^n)}-\log \frac{\vec{p}(Y^n|F^n,X^n)}{\prod_{i=1}^n p(Y_i|Y^{i-1},F^i)}\\
&\stackrel{(b)}{=}\log \frac{\prod_{i=1}^n p(Y_i|X^i,Y^{i-1})}{p(Y^n)}-\log \frac{\vec{p}(Y^n|F^n,X^n)}{\vec{p}(Y^n|F^n)}\\
&=\log \frac{\vec{p}(Y^n|X^n)}{p(Y^n)}-\log \frac{\vec{p}(Y^n|F^n,X^n)}{\vec{p}(Y^n|F^n)}\\
&=i(X^n\rightarrow Y^n)-i(X^n\rightarrow Y^n||F^n) \\
&=i^R(X^n(F^n)\rightarrow Y^n)\\
\end{split}
\end{equation*}
where (a) follows from the fact that no feedback exists from $\mathcal{Y}$ to $\mathcal{F}$. Line (b) follows from the Markov chain $F^i - (X^i,Y^{i-1})- Y_i$.
\end{proof}

\indent In the next lemma, we shows that there exists a suitable construction of $p(f^n)$ such that the induced channel input distribution equals the original channel input distribution. As we will see, this result allows us to work on the channel input distributions instead of code-function distributions.
\begin{lemma}
Given a channel $\lbrace p(y_i|x^i,y^{i-1})\rbrace_{i=1}^n$, a feedback link $\lbrace p(z_i|y^{i},z^{i-1})\rbrace_{i=1}^n$, a channel input distribution $\lbrace p(x_i|x^{i-1},z^{i-1})\rbrace_{i=1}^n$ and a sequence of code-function distributions $\lbrace p(f_i|f^{i-1})\rbrace_{i=1}^n$, the induced channel input distribution $\lbrace p_{ind}(x_i|x^{i-1},z^{i-1})\rbrace_{i=1}^n$ (induced by $\lbrace p(f_i|f^{i-1})\rbrace_{i=1}^n$) equals the original channel input distribution $\lbrace p(x_i|x^{i-1},z^{i-1})\rbrace_{i=1}^n$ if and only if the sequence of code-function distributions $\lbrace p(f_i|f^{i-1})\rbrace_{i=1}^n$ is \textit{good with respect to} $\lbrace p(x_i|x^{i-1},z^{i-1})\rbrace_{i=1}^n$. One choice of such a sequence of code-function distributions is as follows,
\begin{equation}
p(f_i|f^{i-1})=\prod_{z^{i-1}}p(f_i(z^{i-1})|f^{i-1}(z^{i-2}),z^{i-1}).
\label{equa4_6}
\end{equation}
\label{lemma4_2}
\end{lemma}

\indent We refer the readers to Definition $5.1$, Lemma $5.1$ and $5.4$ in \cite{Tati09} for the concept ``\textit{good with respect to}'' and the proof of the above lemma. According to Lemma \ref{lemma4_2}, it is straightforward to obtain the following result which plays an essential role in the channel coding theorem.
\begin{lemma}
For channels with noisy feedback,
\begin{equation*}
\begin{split}
&p(x^n,y^n,f^n)\\
=&\prod_{i=1}^{n} \prod_{z^{i-1}}\underbrace{p(f_i(z^{i-1})|f^{i-1}(z^{i-2}),z^{i-1})}_{\text{Encoding}}\sum_{z^{n}\in \lbrace\mathcal{Z}^{n}:x^n=f^n(z^{n-1})\rbrace}\prod_{i=1}^{n}\underbrace{p(z_i|y^i,z^{i-1})}_{\text{Feedback link}} \underbrace{p(y_i|f^i(z^{i-1}),y^{i-1})}_{\text{Channel}}\\
\end{split}
\end{equation*}
\label{lemma4_6}
\end{lemma}

The proof is shown in the Appendix. This lemma implies that $\underline{I}^R(X(F)\rightarrow Y)$ only depends on channel input distribution $\lbrace p(x_i|x^{i-1},z^{i-1})\rbrace_{i=1}^\infty$.

\subsection{Channel Coding Theorem}
\indent Now we show a general channel coding theorem in terms of the residual directed information.
\begin{theorem}(\textit{Channel Coding Theorem})
For channels with noisy feedback,
\begin{equation}
C_{FB}^{noise}=\sup_{X}\underline{I}^R(X(F)\rightarrow Y)
\label{equ4_10}
\end{equation}
where $\sup_{X}$ means that supremum is taken over all possible channel input distributions $\lbrace p(x_i|x^{i-1},z^{i-1})\rbrace_{i=1}^\infty$.
\label{thm4_2}
\end{theorem}

\indent The proof comes along the proof of Theorem \ref{thm4_1} in \cite{Verdu94} and hence is presented in the Appendix. Theorem \ref{thm4_2} indicates that, besides capturing the effective information flow of channels with noisy feedback, the residual directed information is also beneficial for characterizing the capacity. Although formula (\ref{equ4_10}) may not be the only or the simplest characterization of the noisy feedback capacity, it provides benefits in
many aspects. We herein present two of them as follows.
\begin{enumerate}
\item. Measurements of Information Flows: Let $p^*$ be the optimal solution of formula (\ref{equ4_10}). Then we obtain that, when the channel is used at capacity, the total transmission rate in the forward channel is in fact $\underline{I}(X\rightarrow Y)|_{p^*}$\footnote{$\underline{I}(X\rightarrow Y)|_{p^*}$ denotes that the value is evaluated at channel input distributions $p^*$.} instead of $C_{FB}^{noise}$ and the difference between them (i.e.redundant transmission rate) is $\underline{I}(X\rightarrow Y|F)|_{p^*}$. These numerical knowledge might be crucial in system design and evaluation.
\item. Induced Computable Bounds: Let $q^*= arg\sup_{X}\underline{I}(X\rightarrow Y)$ where supremum is taken over all possible channel input distributions $\lbrace p(x_i|x^{i-1},z^{i-1})\rbrace_{i=1}^\infty$. Since code-function $F$ is not involved at this point, the computation complexity is significantly reduced. Based on Theorem \ref{thm4_2}, it is straightforward to obtain $\underline{I}(X\rightarrow Y)|_{q^*}$ and $\underline{I}^R(X(F)\rightarrow Y)|_{q^*}$ as upper\footnote{Note that $\underline{I}(X\rightarrow Y)|_{q^*}=\sup_{\lbrace p(x_i|x^{i-1},z^{i-1})\rbrace_{i=1}^\infty}\underline{I}(X\rightarrow Y)\leq C_{FB}=\sup_{\lbrace p(x_i|x^{i-1},y^{i-1})\rbrace_{i=1}^\infty}\underline{I}(X\rightarrow Y)$ where $C_{FB}$ is the corresponding perfect feedback capacity. Therefore this upper bound is in general better than $C_{FB}$.} and lower bounds on the capacity, respectively. Further, the gap between the bounds is $\underline{I}(X\rightarrow Y|F)|_{q^*}$, which is definitely a tightness evaluation of the bounds.
\end{enumerate}

\subsection{Computable Bounds on the Capacity}
\indent As it is shown, the capacity characterization in Theorem \ref{thm4_2} is not computable in general due to the probabilistic limit and code-functions. This motivates us to explore some conditions under which the previous characterization can be simplified or to look at some computable bounds instead. Toward this end, we first introduce a strong converse theorem under which the ``probabilistic limit'' can be replaced by the ``normal limit''. We then turn to characterize a pair of upper and lower bounds which is much easier to compute and tight in certain practical situations.
\begin{definition}(\textit{Strong Converse})
A channel with noisy feedback capacity $C_{FB}^{noise}$ has a strong converse if for any $R>C_{FB}^{noise}$, every sequence of channel codes $\lbrace(n,M_n,\epsilon_n)\rbrace_{n=1}^{\infty}$ with
\begin{equation*}
\liminf_{n\rightarrow \infty}\frac{1}{n}\log M_n\geq R
\end{equation*}
satisfies $\lim_{n\rightarrow \infty}\epsilon_n =1$
\end{definition}

\begin{theorem}(\textit{Strong Converse Theorem})
A channel with noisy feedback capacity $C_{FB}^{noise}$ satisfies the strong converse property if and only if
\begin{equation}
\sup_{X}\underline{I}^R(X(F)\rightarrow Y)=\sup_{X}\overline{I}^R(X(F)\rightarrow Y)\footnote{This condition can be alternatively expressed as $\sup_{X}\underline{I}(F;Y)=\sup_{X}\overline{I}(F;Y)$. Since the computation complexity difference between the mutual information and residual directed information is not justified, either condition is a candidate for check. Note that how to check the strong converse is out of the scope of this paper.}
\label{equ4_3}
\end{equation}
Furthermore, if the strong converse property holds, we have
\begin{equation*}
C_{FB}^{noise}=\sup_{X}\lim_{n\rightarrow \infty}\frac{1}{n}I^R(X^n(F^n)\rightarrow Y^n).
\label{thm4_5}
\end{equation*}
\end{theorem}

\indent The proof directly follows from chapter $3.5$ in \cite{bookhan03} by appropriate replacement of $i^R(X^n(F^n)\rightarrow Y^n)$ on $i(F^n; Y^n)$. This theorem gives us an important message that, for channels satisfying the strong converse property, we may compute the noisy feedback capacity by taking the normal limit instead of the probabilistic limit. How to further simplify the capacity characterization will be explored in the future.\\
\indent We next propose computable upper bounds on the noisy feedback capacity.
\begin{theorem}(\textit{Upper Bound})\footnote{As we will see from the proof, this upper bound holds for any finite-alphabet channel with or without strong converse property. }
\begin{equation}
\bar{C}_{FB}^{noise}=\sup_{X}\liminf_{n\rightarrow \infty}\frac{1}{n}I(X^n\rightarrow Y^n||Z^n)
\label{equa4_8}
\end{equation}
where $\bar{C}_{FB}^{noise}$ denotes the upper bound of the capacity and the supremum is taken over all possible channel input distribution $\lbrace p(x_i|x^{i-1},z^{i-1})\rbrace_{i=1}^\infty$.
\label{thm_coding_aditive}
\end{theorem}
\begin{remark}
The computation complexity of formula (\ref{equa4_8}), which is independent of code-functions, is significantly reduced and is similar to that of directed information. We herein conjecture that most of the algorithms for computing the directed information may apply to compute formula (\ref{equa4_8}). For example, for finite-state machine channels \cite{Yang05} with noisy feedback, formula (\ref{equa4_8}) may be computable by using dynamic programming approach along the lines of \cite{Yang05}.
\end{remark}

\indent We need the following lemma before showing the proof of Theorem \ref{thm_coding_aditive}.
\begin{lemma}
\begin{equation*}
I(F^n;Y^n)=I^R(X^n(F^n)\rightarrow Y^n)=I(X^n\rightarrow Y^n||Z^n)-I(F^n;Z^n|Y^n)
\end{equation*}
\label{lemma4_7}
\end{lemma}

\begin{proof}
\begin{equation*}
\begin{split}
&I^R(X^n(F^n)\rightarrow Y^n)\\
\stackrel{(a)}{=}&I(F^n;Y^n)\\
=&I(F^n;(Y^n,Z^n))-I(F^n;Z^n|Y^n) \\
\stackrel{(b)}{=}&I(F^n \rightarrow (Y^n,Z^n))-I(F^n;Z^n|Y^n) \\
=&\sum_{i=1}^n I(F^i,(Y_i,Z_i)|Y^{i-1},Z^{i-1})-I(F^n;Z^n|Y^n)\\
=&\sum_{i=1}^n H(Y_i,Z_i|Y^{i-1},Z^{i-1})-H(Y_i,Z_i|Y^{i-1},Z^{i-1},F^i)-I(F^n;Z^n|Y^n)\\
=&\sum_{i=1}^n H(Z_i|Y^{i},Z^{i-1})+H(Y_i|Y^{i-1},Z^{i-1})-H(Z_i|Y^{i},Z^{i-1},F^i)-H(Y_i|Y^{i-1},Z^{i-1},F^{i})-I(F^n;Z^n|Y^n)  \\
\stackrel{(c)}{=}&\sum_{i=1}^n H(Y_i|Y^{i-1},Z^{i-1})-H(Y_i|Y^{i-1},Z^{i-1},F^{i})-I(F^n;Z^n|Y^n)\\
\stackrel{(d)}{=}&\sum_{i=1}^n H(Y_i|Y^{i-1},Z^{i-1})-H(Y_i|Y^{i-1},X^{i},Z^{i-1},F^{i})-I(F^n;Z^n|Y^n)\\
\stackrel{(e)}{=}&\sum_{i=1}^n H(Y_i|Y^{i-1},Z^{i-1})-H(Y_i|Y^{i-1},X^{i},Z^{i-1})-I(F^n;Z^n|Y^n)\\
=&\sum_{i=1}^n I(X^i,Y_i|Y^{i-1},Z^{i-1})-I(F^n;Z^n|Y^n)\\
=&I(X^n\rightarrow Y^n||Z^{n})-I(F^n;Z^n|Y^n)\\
\end{split}
\end{equation*}
where (a) follows from Lemma \ref{lemma4_3}. Line (b) follows from the fact that there exists no feedback from $(Y^n,Z^n)$ to $F^n$ and thus the mutual information and directed information coincide. Line (c) follows from the fact that $H(Z_i|Y^{i},Z^{i-1})=H(Z_i|Y^{i},Z^{i-1},F^i)$ since $F^i-(Y^{i},Z^{i-1})-Z_i$ forms a Markov chain. Line (d) follows from the fact that $X^i$ can be determined by $F^i$ and the outputs of the feedback link $Z^{i-1}$. Line (e) follows from the Markov chain $F^{i}-(Y^{i-1},X^{i},Z^{i-1})-Y_i$.
\end{proof}

\indent Now we present the proof of Theorem \ref{thm_coding_aditive} as follows.
\begin{proof}
\indent Recall Lemma A1 in \cite{Han93}, we have $\underline{I}(F;Y)\leq \liminf_{n\rightarrow \infty}\frac{1}{n}I(F^n;Y^n)$ for any sequence of joint probability. That is, $\underline{I}^R(X(F)\rightarrow Y)\leq \liminf_{n\rightarrow \infty}\frac{1}{n}I^R(X^n(F^n)\rightarrow Y^n)$. Then by Lemma \ref{lemma4_7},
\begin{equation}
\begin{split}
C_{FB}^{noise}\leq &\sup_{X}\liminf_{n\rightarrow \infty}\frac{1}{n}I^R(X^n(F^n)\rightarrow Y^n)\\
=&\sup_{X}\liminf_{n\rightarrow \infty}\frac{1}{n}(I(X^n\rightarrow Y^n||Z^n)-I(F^n;Z^n|Y^n)) \\
\leq &\sup_{X}\liminf_{n\rightarrow \infty}\frac{1}{n}I(X^n\rightarrow Y^n||Z^n)\\
\end{split}
\label{equa4_7}
\end{equation}
\end{proof}

\begin{corollary}
Assume that there is an independent additive noise feedback (Fig.\ref{figure5}), then
\begin{equation*}
\bar{C}_{FB}^{noise}=\sup_{X}\liminf_{n\rightarrow \infty}\frac{1}{n}I(X^n\rightarrow Y^n|V^n)
\end{equation*}
where $\sup_{X}$ means that supremum is taken over all possible channel input distribution $\lbrace p(x_i|x^{i-1},y^{i-1}+v^{i-1})\rbrace_{i=1}^\infty$.
\label{coro_conditional_directed_info}
\end{corollary}
\begin{proof}
\begin{equation*}
\begin{split}
I(X^n\rightarrow Y^n||Z^n)
=&\sum_{i=1}^n I(X^i,Y_i|Y^{i-1},Z^{i-1})\\
=&\sum_{i=1}^n I(X^i,Y_i|Y^{i-1},V^{i-1})\\
=&I(X^n\rightarrow Y^n||V^n)\\
\stackrel{(a)}{=}&I(X^n\rightarrow Y^n|V^n)\\
\end{split}
\end{equation*}
where (a) follows from remark \ref{remark_pre_01}. The proof is complete.
\end{proof}

\indent Next, we show a lower bound on the capacity for strong converse channels with additive noise feedback. In fact, any particular coding scheme may induce a low bound on the noisy feedback capacity. However, the lower bound proposed in the following has nice features and its own advantages.
\begin{theorem}(\textit{Lower Bound})
Assume that a channel with an independent additive noise feedback (Fig.\ref{figure5}) satisfies the strong converse property. A lower bound on the noisy feedback capacity is given by
\begin{equation*}
\underline{C}_{FB}^{noise}=\bar{C}_{FB}^{noise}-\bar{h}(V)
\end{equation*}
where
\begin{equation*}
\bar{h}(V)=\limsup_{n\rightarrow \infty}\frac{1}{n}H(V^n).
\end{equation*}
\label{coro_conditional_directed_info}
\end{theorem}

\begin{proof}
\indent We need to show that, for any $\delta>0$, there exists a sequence of $(n,M,\epsilon_n)$ channel codes ($\epsilon_n\rightarrow 0$ as $n\rightarrow \infty$) with transmission rate
\begin{equation*}
\begin{split}
R=&\bar{C}_{FB}^{noise}-\bar{h}(V)-\delta\\
=&\sup_{X}\liminf_{n\rightarrow \infty}\frac{1}{n}I(X^n\rightarrow Y^n|V^n)-\bar{h}(V)-\delta.\\
\end{split}
\end{equation*}

\indent Now, for any fixed $\delta>0$, we take $\xi$ satisfying $0<\xi<\delta$ and let $X_\xi$ be a sequence of channel input distributions $\lbrace p(x_i|x^{i-1},z^{i-1})\rbrace_{i=1}^\infty$ satisfying
\begin{equation}
\left(\liminf_{n\rightarrow \infty}\frac{1}{n}I(X^n\rightarrow Y^n||Z^n)\right)\bigg|_{X=X_\xi}=\sup_{X}\liminf_{n\rightarrow \infty}\frac{1}{n}I(X^n\rightarrow Y^n||Z^n)-\xi
\label{equ4_4}
\end{equation}
where $\left(\liminf_{n\rightarrow \infty}\frac{1}{n}I(X^n\rightarrow Y^n||Z^n)\right)\vert_{X=X_\xi}$ denotes that $\liminf_{n\rightarrow \infty}\frac{1}{n}I(X^n\rightarrow Y^n||Z^n)$ is evaluated at $X=X_\xi$. According to the definition of supremum, the existence of $X_\xi$ is guaranteed. Since for strong converse channels we have
\begin{equation*}
C_{FB}^{noise}=\sup_{X}\lim_{n\rightarrow \infty}\frac{1}{n}I^R(X^n(F^n)\rightarrow Y^n),
\end{equation*}
we know that, for any $\delta>0$, there exist a sequence of $(n,M,\epsilon_n)$ channel codes ($\epsilon_n\rightarrow 0 $ as $ n\rightarrow \infty$) with transmission rate
\begin{equation*}
R=\left(\lim_{n\rightarrow \infty}\frac{1}{n}I^R(X^n(F^n)\rightarrow Y^n)\right)\bigg|_{X=X_\xi}-(\delta-\xi).
\end{equation*}
By Lemma \ref{lemma4_7},
\begin{equation*}
\begin{split}
R=&\left(\lim_{n\rightarrow \infty}\frac{1}{n}(I(X^n\rightarrow Y^n||Z^n)-I(F^n;Z^n|Y^n))\right)\bigg|_{X=X_\xi}-(\delta-\xi) \\
=&\left(\lim_{n\rightarrow \infty}\frac{1}{n}(I(X^n\rightarrow Y^n||Z^n)-H(Z^n|Y^n)+H(Z^n|Y^n,F^n))\right)\bigg|_{X=X_\xi}-(\delta-\xi) \\
\geq &\left(\liminf_{n\rightarrow \infty}\frac{1}{n}(I(X^n\rightarrow Y^n||Z^n)-H(Z^n|Y^n))\right)\bigg|_{X=X_\xi}-(\delta-\xi) \\
=&\left(\liminf_{n\rightarrow \infty}\frac{1}{n}(I(X^n\rightarrow Y^n||Z^n)-\sum_{i=0}^n H(Z_i|Z^{i-1},Y^n))\right)\bigg|_{X=X_\xi}-(\delta-\xi) \\
\geq &\left(\liminf_{n\rightarrow \infty}\frac{1}{n}(I(X^n\rightarrow Y^n||Z^n)-\sum_{i=0}^n H(Z_i|Z^{i-1},Y^i))\right)\bigg|_{X=X_\xi}-(\delta-\xi) \\
\stackrel{(a)}{=}&\left(\liminf_{n\rightarrow \infty}\frac{1}{n}(I(X^n\rightarrow Y^n||Z^n)-\sum_{i=0}^n H(V_i|V^{i-1}))\right)\bigg|_{X=X_\xi}-(\delta-\xi) \\
\geq &\left(\liminf_{n\rightarrow \infty}\frac{1}{n}(I(X^n\rightarrow Y^n||Z^n)-H(V^n))\right)\bigg|_{X=X_\xi}-(\delta-\xi) \\
\geq &\left(\liminf_{n\rightarrow \infty}\frac{1}{n}I(X^n\rightarrow Y^n||Z^n)\right)\bigg|_{X=X_\xi}+\liminf_{n\rightarrow \infty}-\frac{1}{n}H(V^n)-(\delta-\xi) \\
=&\left(\liminf_{n\rightarrow \infty}\frac{1}{n}I(X^n\rightarrow Y^n||Z^n)\right)\bigg|_{X=X_\xi}-\limsup_{n\rightarrow \infty}\frac{1}{n}H(V^n)-(\delta-\xi) \\
\stackrel{(b)}{=}&\sup_{X}\liminf_{n\rightarrow \infty}\frac{1}{n}I(X^n\rightarrow Y^n||Z^n)-\xi-\bar{h}(V)-(\delta-\xi)  \\
=&\sup_{X}\liminf_{n\rightarrow \infty}\frac{1}{n}I(X^n\rightarrow Y^n||Z^n)-\bar{h}(V)-\delta\\
\stackrel{(c)}{=}&\sup_{X}\liminf_{n\rightarrow \infty}\frac{1}{n}I(X^n\rightarrow Y^n|V^n)-\bar{h}(V)-\delta\\
\end{split}
\end{equation*}
where (a) follows from the fact that $Z_i=Y_i+V_i$ and the Markov Chain $(Z^{i-1},Y^{i})-V^{i-1}-V_i$. Line (b) follows from equation (\ref{equ4_4}). Line (c) follows from Corollary $3$.\\
\indent Since $\delta$ can be arbitrarily small, the proof is complete.
\end{proof}

\begin{remark}
This theorem reveals an important message that the gap between the proposed upper and lower bounds only depends on the feedback additive noise $V$ (i.e. independent from the forward channel). Further, if the entropy rate of noise $V$ goes to zero\footnote{In many practical situations, the entropy rate of the feedback noise is small. For example, if the feedback link only suffers intersymbol interference as illustrated in Chapter $4$ \cite{bookGallager68}, the entropy rate turns out to be approximately $0.0808$. Further, if the cardinality of $V^\infty$ is finite (yet the feedback is still noisy), the entropy rate is clearly zero.}, the proposed upper and lower bound converges and thus the capacity is known.
\end{remark}

\indent We end this section by investigating two examples of noisy feedback channels.

\begin{example}
\indent The example shows that for DMC with noisy feedback the characterized upper bound equals to the open-loop capacity. This implies that the upper bound should be tight when the channel ``converges'' to DMC. Besides, this example verifies the result (i.e. Theorem \ref{Thm_4_1}) in Section IV.\\
\indent Consider a binary symmetric channel (BSC) with a binary symmetric feedback. Note that this is the simplest model of a noisy feedback channel, yet it captures most features of the general problem. We model the noisy channel/feedback as additive noise channel/feedback as follows.
\begin{equation*}
Y_i=X_i+ U_i  \quad (\text{mod}\quad 2) \quad \text{and} \quad Z_i=Y_i+ V_i \quad (\text{mod} \quad 2)
\end{equation*}
where we assume that $Pr(U_i=1)=1-Pr(U_i=0)=\alpha$ and $Pr(V_i=1)=1-Pr(V_i=0)=\beta$. It is known that the capacity of this noisy feedback channel equals the nonfeedback capacity $1-H(\alpha)$ where $H(\alpha)=-\alpha\log{\alpha}-(1-\alpha)\log{(1-\alpha)}$. Next, we show that maximizing the conditional directed information in Corollary \ref{coro_conditional_directed_info} provides the noisy feedback capacity. That is,
\begin{equation*}
\sup_{X}\liminf_{n\rightarrow \infty}\frac{1}{n}I(X^n\rightarrow Y^n|V^n)=1-H(\alpha).
\end{equation*}

\indent This can be done as follows.
\begin{equation*}
\begin{split}
\liminf_{n\rightarrow \infty}\frac{1}{n}I(X^n\rightarrow Y^n|V^n)=&\liminf_{n\rightarrow \infty}\frac{1}{n}\sum_{i=1}^n I(X^i;Y_i|Y^{i-1},V^{i-1})\\
=&\liminf_{n\rightarrow \infty}\frac{1}{n}\sum_{i=1}^n H(Y_i|Y^{i-1},V^{i-1})-H(Y_i|X^{i},Y^{i-1},V^{i-1}) \\
=&\liminf_{n\rightarrow \infty}\frac{1}{n}\sum_{i=1}^n H(Y_i|Y^{i-1},V^{i-1})-H(Y_i|X_{i}) \\
=&\liminf_{n\rightarrow \infty}\frac{1}{n}\sum_{i=1}^n H(Y_i|Y^{i-1},V^{i-1})-H(U_i) \\
\stackrel{(a)}{\leq} &\liminf_{n\rightarrow \infty}\frac{1}{n}\sum_{i=1}^n H(Y_i|Y^{i-1})-H(U_i) \\
\stackrel{(b)}{\leq} &\liminf_{n\rightarrow \infty}\frac{1}{n}\sum_{i=1}^n H(Y_i)-H(U_i)\\
\stackrel{(c)}{\leq} &1-H(\alpha)\\
\end{split}
\end{equation*}
where taking equality in $(a)$ implies $\liminf_{n\rightarrow \infty}\frac{1}{n}I(V^{n-1};Y^n)=0$, that is, the capacity-achieving encoder should be non-typical. This verifies Theorem \ref{Thm_4_1} in Section IV. Taking equalities in $(b)$ and $(c)$ imply that the capacity-achieving encoder should produce equal-probability channel outputs (i.e. uniform distribution). It is obvious that there exists such an optimal encoder that all above equalities hold.
\end{example}

\begin{figure}
\begin{center}
\includegraphics[scale=0.75]{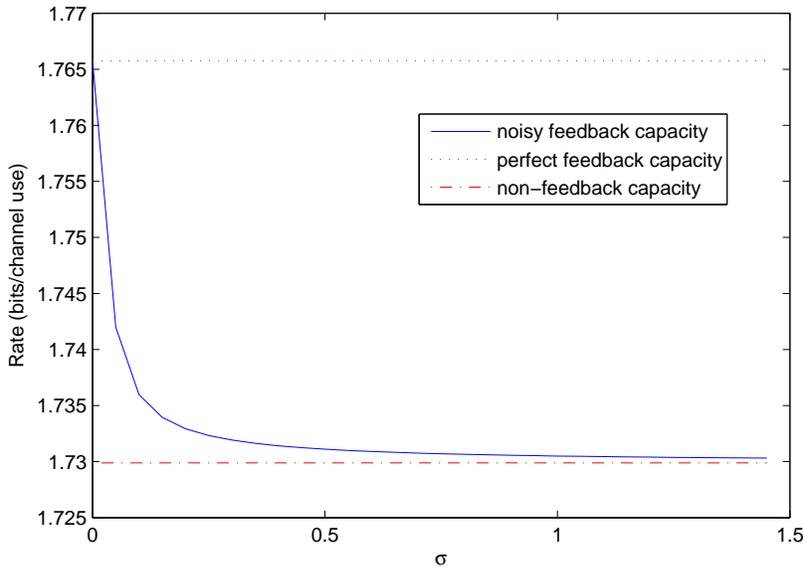}
\caption{The upper bound on the capacity of a first moving average Gaussian channel with AWGN feedback.}
\label{sim_fig01_0.1}
\end{center}
\end{figure}

\begin{example} In this example, we consider a colored Gaussian channel with additive white Gaussian noise feedback and compute the proposed upper bound\footnote{Although the Gaussian channels are not finite-alphabet, the upper bound characterization still holds. The derivation of the upper bound follows exactly the same idea in this paper and can be found in \cite{Chong11_allerton}.}. Specifically, we assume the forward channel and the feedback link as follows.
\begin{equation*}
Y_i=X_i+W_i \quad \text{and} \quad Z_i=Y_i+V_i
\end{equation*}
where $W_i=U_i+0.1U_{i-1}$, $U_i$ is a white Gaussian process with zero mean and unit variance and $V_i$ a white Gaussian process with zero mean and variance $\sigma$. We take coding block length $n=30$ and power limit $P=10$ for computing the upper bound. See Fig. \ref{sim_fig01_0.1}. We refer the interested readers to \cite{Chong11_isit_bounds, Chong11_allerton} for the details of the computation and discussions. From the plot of the upper bound, we see that the noisy feedback capacity is very sensitive to the feedback noise, at least for certain Gaussian channels.
\end{example}

\section{Conclusion}
\indent We proposed a new concept, the \textit{residual directed information} for characterizing the effective information flow through communication channels with noisy feedback, which extends Massey's concept of \textit{directed information}. Based on this new concept, we first analyzed the information flow in noisy feedback channels and then showed that the capacity of DMC is not achievable by using any typical closed-loop encoder. We next proved a general channel coding theorem in terms of the proposed residual directed information. Finally, we proposed computable bounds characterized by the causal conditional directed information.\\
\indent The results in the paper open up new directions for investigating the role of noisy feedback in communication systems. Furthermore, the new definitions, concepts and methodologies presented in the paper are potential to be extended to multiple access channels, broadcast channels or general multi-user channels with noisy feedback.

\section{Appendix}
\subsection{Proof of Lemma \ref{lemma4_6}}
Before giving the proof, we need the following Lemma.
\begin{lemma}
For channels with noisy feedback, as shown in Fig.\ref{figure1},
\begin{equation*}
p(x^n,y^n)=\sum_{z^{n}\in \mathcal{Z}^{n}}\prod_{i=1}^{n}\underbrace{p(z_i|y^i,z^{i-1})}_{\text{Feedback link}} \underbrace{p(x_i|x^{i-1},z^{i-1})}_{\text{Encoding}}\underbrace{p(y_i|x^i,y^{i-1})}_{\text{Channel}}
\end{equation*}
\label{lemma4_0}
\end{lemma}
\begin{proof}
\begin{equation*}
\begin{split}
p(x^n,y^n)&=\sum_{z^{n}\in \mathcal{Z}^{n}}p(x^n,y^n,z^n)\\
&=\sum_{z^{n}\in \mathcal{Z}^{n}}p(z_{n}|x^{n},y^{n},z^{n-1}) p(x^{n},y^{n},z^{n-1})\\
&=\sum_{z^{n}\in \mathcal{Z}^{n}}p(z_{n}|x^{n},y^{n},z^{n-1}) p(y_n|x^{n},y^{n-1},z^{n-1}) p(x^{n},y^{n-1},z^{n-1})\\
&=\sum_{z^{n}\in \mathcal{Z}^{n}}p(z_{n}|x^{n},y^{n},z^{n-1}) p(y_n|x^{n},y^{n-1},z^{n-1}) p(x_n|x^{n-1},y^{n-1},z^{n-1})\\
&\indent  p(x^{n-1},y^{n-1},z^{n-1})\\
&\stackrel{(a)}{=}\sum_{z^{n}\in \mathcal{Z}^{n}}p(z_{n}|y^{n},z^{n-1}) p(y_n|x^{n},y^{n-1}) p(x_n|x^{n-1},z^{n-1})\\
&\indent  p(x^{n-1},y^{n-1},z^{n-1}) \\
&=\sum_{z^{n}\in \mathcal{Z}^{n}}\prod_{i=1}^{n}p(z_i|y^i,z^{i-1}) p(x_i|x^{i-1},z^{i-1})p(y_i|x^i,y^{i-1}) \\
\end{split}
\end{equation*}
where $(a)$ follows from the Markov chains: $x^{n}- (y^{n},z^{n-1})-z_{n}$, $z^{n-1}- (x^{n},y^{n-1})-y_{n}$ and $y^{n-1}-(x^{n-1},z^{n-1})-x_{n}$.
\end{proof}

\indent Now, we are ready to give the proof of Lemma \ref{lemma4_6}.
\begin{proof}
\begin{equation*}
\begin{split}
\indent &p(x^n,y^n,f^n)\\
&=p(x^n,y^n|f^n) p(f^n)\\
&\stackrel{(a)}{=}p(f^n) \sum_{z^{n}\in \mathcal{Z}^{n}}\prod_{i=1}^{n}p(z_i|y^i,z^{i-1},f^n) p(x_i|x^{i-1},z^{i-1},f^n) p(y_i|x^i,y^{i-1},f^n)\\
&=p(f^n) \sum_{z^{n}\in \lbrace\mathcal{Z}^{n}:x^n=f^n(z^{n-1})\rbrace}\prod_{i=1}^{n}p(z_i|y^i,z^{i-1},f^n) p(y_i|f^i(z^{i-1}),y^{i-1},f^n) \\
&\stackrel{(b)}{=}p(f^n) \sum_{z^{n}\in \lbrace\mathcal{Z}^{n}:x^n=f^n(z^{n-1})\rbrace}\prod_{i=1}^{n}p(z_i|y^i,z^{i-1}) p(y_i|f^i(z^{i-1}),y^{i-1}) \\
&\stackrel{(c)}{=}\prod_{i=1}^{n} \prod_{z^{i-1}}p(f_i(z^{i-1})|f^{i-1}(z^{i-2}),z^{i-1}) \sum_{z^{n}\in \lbrace\mathcal{Z}^{n}:x^n=f^n(z^{n-1})\rbrace}\prod_{i=1}^{n}p(z_i|y^i,z^{i-1})p(y_i|f^i(z^{i-1}),y^{i-1}) \\
\end{split}
\end{equation*}
where (a) follows from Lemma \ref{lemma4_0}. Line (b) follows from the Markov chains: $f^n - (y^i,z^{i-1})- z_i$ and $f^n -(f^i(z^{i-1}),y^{i-1})-y_i$. Line (c) follows from Lemma \ref{lemma4_2}.
\end{proof}

\subsection{Proof of Theorem \ref{thm4_2}}
\indent We now prove the channel coding theorem by combining the following converse theorem and achievability theorem.\\
a). \textit{Converse Theorem}\\
\indent The following is a generalization of theorem $4$ in \cite{Verdu94} which gives an upper bound for bounding the block error probability.
\begin{lemma}
Every $(n,M,\epsilon)$ channel code satisfies
\begin{equation*}
\epsilon\geq Prob\lbrace\frac{1}{n}i^R(X^n(F^n)\rightarrow Y^n)\leq \frac{1}{n}\log M-\gamma\rbrace-2^{-\gamma n}
\end{equation*}
for every $\gamma>0$.
\label{lemma4_4}
\end{lemma}

\begin{proof}
We assume the disjointness of the decoding sets $D$. i.e. $D_{w=i}\cap D_{w=j}=\emptyset$ if $i \neq j$. Under this restriction on the decoder, \cite{Verdu94} has shown that any $(n,M,\epsilon)$ channel code for the nonfeedback channel $\mathcal{F}^n\rightarrow \mathcal{Y}^n$ satisfies for all $\gamma>0$
\begin{equation*}
\epsilon\geq Prob\lbrace\frac{1}{n}i(F^n; Y^n)\leq \frac{1}{n}\log M-\gamma \rbrace-2^{-\gamma n}
\end{equation*}
By Lemma \ref{lemma4_3}, we have
\begin{equation*}
i(F^n;Y^n)=i^R(X^n(F^n)\rightarrow Y^n)
\end{equation*}
The proof is complete.
\end{proof}

\indent Note that this Lemma holds independently of the decoder that one uses. The only restriction on the decoder is the disjointness of the decoding region.
\begin{theorem}(\textit{Converse Theorem})
\begin{equation*}
C_{FB}^{noise}\leq \sup_{X}\underline{I}^R(X(F)\rightarrow Y)
\end{equation*}
\label{thm4_3}
\end{theorem}

\begin{proof}
Assume that there exists a sequence of $(n,M,\epsilon_n)$ channel codes with $\epsilon_n\rightarrow 0$ as $n\rightarrow \infty$ and with transmission rate
\begin{equation*}
R=\liminf_{n\rightarrow \infty}\frac{1}{n}\log M.
\end{equation*}
By Lemma \ref{lemma4_4}, we know that for all $\gamma>0$,
\begin{equation*}
\epsilon_n \geq Prob\lbrace\frac{1}{n}i^R(X^n(F^n)\rightarrow Y^n)\leq \frac{1}{n}\log M-\gamma\rbrace-2^{-\gamma n}
\end{equation*}
As $n\rightarrow \infty$, the probability on the right-hand side must go to zero since $\epsilon_n\rightarrow 0$.
By the definition of $\underline{I}^R(X(F)\rightarrow Y)$, we have
\begin{equation*}
\limsup_{n\rightarrow \infty}\frac{1}{n}\log M-\gamma \leq \underline{I}^R(X(F)\rightarrow Y)
\end{equation*}
Since $\gamma$ can be arbitrarily small, we have
\begin{equation*}
R\leq \limsup_{n\rightarrow \infty}\frac{1}{n}\log M\leq \underline{I}^R(X(F)\rightarrow Y)\leq \sup_{X} \underline{I}^R(X(F)\rightarrow Y)
\end{equation*}
The proof is complete.
\end{proof}
%%%%%%%%%%%%%%%%%%%%%%%%%%%%%%%%%%%%%%%%%%%%%%%%%%%%%%%%%%%%%%%%%%%%%%%%%%%%%%%%%%%%%%%%%%%%%%%%%%%%%%%%%%%%%%%%%%%%%%%%%%%%%%%%%%%%%%%%%%%%%%%%%%%%%%
b). \textit{Achievability Theorem}\\
\indent The following is a generalization of Feinstein's lemma \cite{Feinstein54} based on the residual directed information.
\begin{lemma}
Fix a positive integer $n$, $0<\epsilon<1$, a channel $\lbrace p(y_i|x^i,y^{i-1})\rbrace_{i=1}^n$ and a feedback link $\lbrace p(z_i|y^i,z^{i-1})\rbrace_{i=1}^n$. For every $\gamma>0$ and a channel input distribution $\lbrace p(x_i|x^{i-1},z^{i-1})\rbrace_{i=1}^n$, there exists a channel code $(n,M,\epsilon)$ that satisfies
\begin{equation*}
\epsilon \leq Prob\lbrace\frac{1}{n}i^R(X^n(F^n)\rightarrow Y^n)\leq \frac{1}{n}\log M+\gamma\rbrace+2^{-\gamma n}
\end{equation*}
\label{lemma4_5}
\end{lemma}
\begin{proof}
Given a channel input distribution $\lbrace p(x_i|x^{i-1},z^{i-1})\rbrace_{i=1}^n$, we generate a code-function distribution $\lbrace p(f_i|f^{i-1})\rbrace_{i=1}^n$ such that the induced channel input distribution equals the original channel input distribution. There exists such a code-function distribution according to Lemma \ref{lemma4_2}. In \cite{Verdu94}, it has been shown that for a nonfeedback channel $\lbrace p(y_i|f^i,y^{i-1})\rbrace_{i=1}^n$, a channel input distribution $\lbrace p(f_i|f^{i-1})\rbrace_{i=1}^n$ and for every $\gamma>0$, there exists a channel code $(n,M,\epsilon)$ that satisfies
\begin{equation*}
\epsilon \leq Prob\lbrace\frac{1}{n}i(F^n; Y^n)\leq \frac{1}{n}\log M+\gamma\rbrace+2^{-\gamma n}
\end{equation*}
Recall that this result is proved by random coding argument. Then, by Lemma \ref{lemma4_3}, we have
\begin{equation*}
i(F^n;Y^n)=i^R(X^n(F^n)\rightarrow Y^n)
\end{equation*}
The proof is complete after simple replacement.
\end{proof}

\begin{theorem}(\textit{Achievability Theorem})
\begin{equation*}
C_{FB}^{noise}\geq \sup_{X}\underline{I}^R(X(F)\rightarrow Y)
\end{equation*}
\label{thm4_4}
\end{theorem}
\begin{proof}
Fix arbitrary $0<\epsilon<1$ and channel input distribution $\lbrace p(x_i|x^{i-1},z^{i-1})\rbrace_{i=1}^n$. We shall show that $\underline{I}^R(X(F)\rightarrow Y)$ is a $\epsilon$-achievable rate by demonstrating that, for every $\delta>0$, and all sufficient large $n$, there exists a $(n,M,2^{-\frac{n\delta}{4}}+\frac{\epsilon}{2})$ code with rate
\begin{equation*}
\underline{I}^R(X(F)\rightarrow Y)-\delta<\frac{\log M}{n}<\underline{I}^R(X(F)\rightarrow Y)-\frac{\delta}{2}
\end{equation*}
If, in Lemma \ref{lemma4_5}, we choose $\gamma=\frac{\delta}{4}$, then the right-hand side value in Lemma \ref{lemma4_5} becomes
\begin{equation*}
\begin{split}
&Prob\lbrace\frac{1}{n}i^R(X^n(F^n)\rightarrow Y^n)\leq \frac{1}{n}\log M+\frac{\delta}{4}\rbrace+2^{-\frac{n\delta}{4}}\\
\leq&Prob\lbrace\frac{1}{n}i^R(X^n(F^n)\rightarrow Y^n)\leq \underline{I}^R(X(F)\rightarrow Y)-\frac{\delta}{4}\rbrace+2^{-\frac{n\delta}{4}}\\
\leq&\frac{\epsilon}{2}+2^{-\frac{n\delta}{4}}\\
\end{split}
\end{equation*}
where the second inequality holds for all sufficiently large n because of the definition of $\underline{I}^R(X(F)\rightarrow Y)$. Therefore, $\underline{I}^R(X(F)\rightarrow Y)$ is a $\epsilon$-achievable rate.
\end{proof}

\indent The proof of Theorem \ref{thm4_2} is obtained by combining Theorem \ref{thm4_3} and Theorem \ref{thm4_4}. 
\bibliographystyle{IEEEtran}
\bibliography{ref}

\end{document}